\documentclass[letterpaper,final,twocolumn]{IEEEtran}

\usepackage{graphicx,subfigure}
\usepackage{amsmath,amssymb,amstext,amsfonts,mathtools}
\usepackage{epsfig}
\usepackage{psfrag}
\usepackage{stfloats,color}
\usepackage{fancyhdr}
\usepackage[T1]{fontenc}
\usepackage{nccmath}
\usepackage{algorithm,algorithmic}
\usepackage{wasysym,dblaccnt,accents}
\usepackage{epstopdf}
\usepackage{leftidx}
\usepackage{cite}

\graphicspath{{figures/}}

\newtheorem{theorem}{Theorem}[section]

\newtheorem{assumption}[theorem]{Assumption}
\newtheorem{problem}[theorem]{Problem}
\newtheorem{lemma}[theorem]{Lemma}

\newtheorem{remark}[theorem]{Remark}
\newtheorem{example}[theorem]{Example}
\newtheorem{corollary}[theorem]{Corollary}
\newtheorem{proposition}[theorem]{Proposition}

\newcommand{\longthmtitle}[1]{\mbox{}{\bf \textit{(#1).}}}

\newcommand{\real}{\ensuremath{\mathbb{R}}}
\newcommand{\realpos}{\ensuremath{\mathbb{R}_{>0}}}

\newcommand{\intnonneg}{{\mathbb{Z}}_{\ge 0}}

\newcommand{\setdef}[2]{\{#1 \; | \; #2\}}

\newcommand{\E}{\mathbb{E}}

\newcommand{\until}[1]{\{1,\dots,#1\}}

\newcommand{\norm}[1]{\lVert #1 \rVert}

\newcommand{\Ac}{\mathcal{A}}

\newcommand{\Nc}{\mathcal{N}}

\newcommand{\oprocendsymbol}{\hbox{$\square$}}
\newcommand{\oprocend}{\relax\ifmmode\else\unskip\hfill\fi\oprocendsymbol}

\allowdisplaybreaks
\emergencystretch=\maxdimen
\input{tcilatex}

\newcommand\rev[1]{{\color{blue} #1}}
\renewcommand\rev[1]{#1}
\newcommand\nrev[1]{{\color{blue} #1}}
\renewcommand\nrev[1]{#1}
\newcommand\nnrev[1]{{\color{blue} #1}}
\renewcommand\nnrev[1]{#1}

\begin{document}

\title{Network Identification with Latent Nodes
  \\
  via Auto-Regressive Models\thanks{A preliminary version 
    appeared as~\cite{YZ-JC:16-acc} at the American Control
    Conference.}}

\author{Erfan Nozari \quad Yingbo Zhao \quad Jorge Cort\'{e}s
  \thanks{%
    E. Nozari and J. Cort\'es are with the Department of Mechanical
    and Aerospace Engineering, University of California at San Diego,
    La Jolla, CA 92093, \texttt{\small
      \{enozari,cortes\}@ucsd.edu}. Y. Zhao is with Cymer Corporation,
    San Diego, CA 92127, \texttt{\small yingbo.zhao@asml.com}}}

\maketitle

\begin{abstract}
  We consider linear time-invariant networks with unknown topology
  where only a \emph{manifest} subset of the nodes can be directly
  \nrev{actuated and measured} while the state of the remaining
  \emph{latent} nodes and their number are unknown. Our goal is to
  identify the transfer function of the manifest subnetwork and
  determine whether interactions between manifest nodes are direct or
  mediated by latent nodes.  We show that, if there are no inputs to
  the latent nodes, the manifest transfer function can be approximated
  arbitrarily well in the $H_{\infty}$-norm sense by the transfer
  function of an auto-regressive model and present a least-squares
  estimation method to construct the auto-regressive model from
  measured data. We show that the least-squares auto-regressive method
  guarantees an arbitrarily small $H_\infty$-norm error in the
  approximation of the manifest transfer function, exponentially
  decaying once the model order exceeds a certain threshold.  Finally,
  we show that when the latent subnetwork is acyclic, the proposed
  method achieves perfect identification of the manifest transfer
  function above a specific model order as the length of the data
  increases.  Various examples illustrate our results.
\end{abstract}

\section{Introduction}\label{Sec:Intro}

Network reconstruction problems are widespread in many areas of
science and engineering.  In systems biology, for instance, genetic
network identification uses data from RNA micro-array experiments to
identify the interaction pattern between genes in a \nrev{regulatory
  network~\cite{YXRW-HH:14,AJ-MZ-SB-GJP:09}.} In neuroscience,
researchers seek to understand how different regions of the brain
cooperate with each other by having subjects \nrev{perform certain
  goal-directed tasks while measuring their brain activity via
  multi-channel recordings such as electroencephalograms
  (EEG)~\cite{VS:11,SLB-AKS:11,JRI-AO-TM-MP-JS-GC-HP:14,AK-CMC-RK-PJF-NEC:08,MK-MD-WAT-SLB:01}.}
Similar examples exist in other areas including finance, social
networks, and physics.  Roughly speaking, the objective in network
identification is to determine causal relationships among the nodes in
the network that model the direction and strength of the interactions
between them.  While network control and coordination has made much
progress on problems where the interaction topology is either given or
the design objective itself, not so much attention has been devoted to
develop techniques to address the identification of unknown topologies
from measured data. The need for the latter is especially apparent in
the context of complex, large-scale networks, where it is often not
possible to \nrev{measure or actuate} all nodes, or even know their number.
In this paper, we seek to contribute to this body of work by studying
the effect that the presence of \nrev{unmeasured} nodes has on the
identification of networked linear systems with arbitrary topology.

\subsubsection*{Literature review}
An increasing number of works study topology identification problems
to better understand the interactions in large-scale networks and
their role in determining the network behavior.  A complex network is
commonly represented as a directed graph, and the interactions among
neighboring nodes are represented by directed edges whose weights
reflect the interaction strength. In this sense, topology
identification aims at identifying the adjacency matrix of the network
graph~\cite{CDG-GFR:01} \rev{or its Boolean
  structure~\cite{JS-DT-EMB:15}.}  The work~\cite{MNA-MM:12} studies
the complete characterization of the interaction topology of
consensus-type networks using a series of node-knockout experiments,
where nodes are sequentially forced to broadcast a zero state without
being removed from the network.  The work~\cite{SS-VMP:15} also uses
node-knockout experiments to identify the topology of directed linear
time-invariant networks relying on the cross-power spectral densities
of the network response to wide-sense stationary noise.  The
work~\cite{MT:07} presents a method to infer the topology of a network
of coupled phase oscillators from its stable response dynamics,
assuming that one can manipulate every individual node and perform
large number of experiments. In general, without such assumption, it
is difficult \rev{or impossible, depending on the additional
  structural information available,} to accurately identify the
topology of a general network.
As a result, a main focus has been on particular network realizations
that explain the measured data, such as the sparsest realization,
sometimes with a design parameter to manage the trade-off between
model accuracy and sparsity, see
e.g.,~\cite{AJ-MZ-SB-GJP:09,DM-GI-LG-MS:13}. Along these lines, the
work~\cite{DM-GI:10} considers the identification of networked linear
systems with tree topologies.  The above-referenced works rely on
knowledge of the number of nodes in the network. However, it is often
impossible to sample the state of all nodes, or even know the
existence of some of them.  The work~\cite{MJC-VYT-AA-ASW:11} studies
the problem of learning latent tree graphical models where samples are
available only from a subset of the nodes, and proposes
computationally efficient algorithms for learning trees without any
redundant hidden nodes.  The work~\cite{DM-MVS:12} proposes a method
to identify the latent nodes and consistently reconstruct the topology
under the assumptions that the network is a polytree and the degree of
each latent node is at least three, with out-degree of at least two.
Unlike the topology identification algorithms proposed
in~\cite{DM-GI:10,DM-MVS:12}, our approach here allows for the
possibility of cycles in the network topology.  \nrev{%
  Using the notion of the dynamical structure function of a network
  with latent nodes~\cite{JG-SW:08}, the work~\cite{YY-GBS-SW-JG:11}
  proposes a convex optimization-based approach to find the best
  Boolean structure among manifest nodes which consists of computing
  and comparing the distance between an estimated transfer matrix or
  data to all possible Boolean structures. The problem of minimal
  state-space realization of a given dynamical structure function was
  further studied in~\cite{YY-KG-JG:15}. In the present work, however,
  we use a least-square autoregressive identification approach to
  identify not only whether a pair of manifest nodes are dynamically
  connected, but also whether this connection is direct or indirect
  (latent-mediated) and, in the latter case, the length of the
  shortest path between the two.  }%
\rev{Recent work has employed sparse plus low-rank (S+L) decomposition to
  identify general graphical models (with the possibility of cycles)
  with latent variables for static~\cite{VC-PAP-ASW:12} and
  dynamic~\cite{MZ-RS:16} models. The present paper has two main
  differences with respect to this work.  First, the S+L decomposition
  assumes that the subnetwork among manifest nodes is sparse and the
  number of latent nodes is (considerably) smaller than the number of
  manifest ones, while our method is applicable to arbitrary
  networks. Second, although the identification procedure
  of~\cite{MZ-RS:16} also leads to an auto-regressive (AR) model, it
  is based on the so-called maximum-entropy covariance extension. This
  method, with origins in seismic vibrations and human voice analysis,
  seeks to \emph{maximize} the prediction error~\cite{CIB-SVG-AL:01}
  (while our approach seeks to \emph{minimize} it), leading to very
  different models.}
Finally, our work is inspired by the wide use in neuroscience of AR
models to analyze brain data via Granger causality and its variants
and the study of effective connectivity among different areas of the
brain, see
e.g.,~\cite{EB-OS:09,SLB-AKS:11,JRI-AO-TM-MP-JS-GC-HP:14}. The Granger
causality measure is a mainly descriptive tool that captures influence
and interconnection among time series. \nrev{A popular variant of
  Granger causality, direct directed transfer function
  (dDTF)~\cite{MK-MD-WAT-SLB:01,AK-MM-MK-KJB-SK:03} distinguishes
  between direct and indirect interconnections between two nodes by
  multiplying the directed transfer function (DTF, the normalized
  transfer function between the two nodes) by the partial coherence
  between them in the frequency domain.} We are motivated here by
understanding to what extent the reconstruction results obtained via
methods that build on Granger causality are sensitive to the presence
of latent nodes. \nrev{Furthermore, we propose a method using
  (multivariate) AR models for networks with latent nodes that
  distinguishes between direct and indirect (i.e., latent-mediated)
  interconnections between two nodes in the time domain based on the
  order of the interconnection between them.}

\subsubsection*{Statement of contributions}

We consider a scenario where one can only directly \nrev{actuate and measure}
a subset of the nodes, termed manifest, of a large linear
time-invariant network whose total number of nodes and interaction
topology are unknown.  The objective is to identify the manifest
transfer function, which is the submatrix corresponding to the
manifest nodes of the transfer function matrix of the entire
network. To achieve this, we study the transfer functions provided by
linear AR models. \rev{Our discussion shows how AR models can be used
  to effectively
  distinguish direct interactions between manifest nodes from indirect
  interactions mediated by latent nodes.}
Our first contribution shows that, if no inputs act on the latent
nodes, then there exists a class of AR models whose transfer functions
converge exponentially in the $H_{\infty }$ norm to the manifest
transfer function as the model order increases. We also show that, if
the latent subnetwork is acyclic, then this approximation is exact
above a specific model order.  Our second contribution characterizes
the properties of using least-squares auto-regressive estimation to
construct the AR model from measured data.  We establish that the
least-squares matrix estimate converges in probability to the optimal
matrix sequence identified in our first contribution, enabling us to
determine whether two manifest nodes interact directly or indirectly
through latent nodes. We also show that the least-squares
auto-regressive method guarantees an arbitrarily small $H_{\infty
}$-norm error as the length of data and the model order grow. In fact,
once the order of the AR model candidates exceeds a certain threshold,
the $H_{\infty }$-norm error decays exponentially. Finally, we show
that, when the latent subnetwork is acyclic, the method achieves
perfect identification of the manifest transfer function.
\rev{Throughout a series of remarks in the paper, we also discuss how
  our results can be extended to the identification of linear network
  models of arbitrary order. Simulations on a directed ring network,
  Erd\H{o}s--R{\'{e}}nyi random graphs, and real EEG data illustrate
  our results.}

\subsubsection*{Notation}

For a vector $x\in \mathbb{R}^n$, we use $x_{i}$ to denote its $i$-th
element. Given a sequence $\{x(k)\}_{k=0}^{\infty } \subset
\mathbb{R}^n$ and $j_{1}\leq j_{2}\in \mathbb{Z}_{\geq 0}$, we use
$\{x\}_{j_{1}}^{j_{2}}$ to denote the finite sequence
$\{x(j_{1}),x(j_{1}+1),\dotsc ,x(j_{2})\}$. We omit $j_{1}$ if
$j_{1}=0$. We denote $\norm{\{x\}_{j_{1}}^{j_{2}}} \triangleq \bigl(
\sum_{k=j_{1}}^{j_{2}}x^{T}(k)x(k)\bigr)^{\frac{1}{2}}$. A sequence of
random variables $\{x\}$ converges in probability to a random variable
$X$, denoted $\func{plim} _{k\rightarrow \infty } x(k) = X$, if
$\lim_{k\rightarrow \infty }\Pr (|x(k)-X| \geq \varepsilon )=0$ for
all $\varepsilon >0$. \nnrev{Accordingly, a sequence of random
  matrices $\{A\}$ converges to a random matrix $A_\infty$ in
  probability if $\func{plim}_{k \to \infty} A_{ij}(k) = A_{\infty,
    ij}$ for all $i, j$.}  For a real matrix $M\in \mathbb{R}^{m\times
  n}$, we denote its singular values in decreasing order as $\sigma
_{1}(M)\geq \sigma _{2}(M)\geq \cdots $ $\geq \sigma _{\min
  (m,n)}(M)\geq 0$ and its spectral norm by $\norm{M} =\sigma
_{1}(M)$. The max norm of $M$ is $\norm{M}_{\max}=\max_{i,j}\left\vert
  M_{ij}\right\vert $. We denote by $\rho (M)$ the spectral radius of
a square matrix~$M$. The matrix $M$ is Schur stable if and only if
$\rho (M)<1$. We let $\mathbf{0}_{m\times n}$ denote the $m\times n$
matrix with all zero elements and by $I_{n}$ the identity matrix of
dimension $n\times n$. The $H_{\infty }$-norm of a discrete transfer
function $T$ is $\norm{T}_{\infty }\triangleq \sup_{-\pi \leq \omega
  \leq \pi }\lVert T(\omega )\rVert $.

\section{Problem formulation}\label{Sec:PF}

We consider a discrete-time, linear time-invariant (LTI) network
dynamics with state-space representation%
\begin{align}\label{eq:LTI_Network}
  x(k+1)& =Ax(k)+u(k), \notag
  \\
  y(k)& =Cx(k), 
\end{align}%
where $k\in \mathbb{Z}_{\geq 0}$ is the time index, $x(k)\in
\mathbb{R}^{n}$ is the network state (with $x_{i}(k)$ representing the
state of node $i \in \until{n}$), $u(k)\in \mathbb{R}^{n}$ is the
control input (with $u_{i}(k)$ acting on node $i$), and $y(k)\in
\mathbb{R}^{m}$ is the network output. Here, $A\in \mathbb{R}%
^{n\times n}$ is the adjacency matrix of the network, characterizing
the interactions among neighboring nodes, and $C\in
\mathbb{R}^{m\times n}$ is the \nrev{output matrix}.
\nrev{Since natural system are usually driven by noise, the input, state, and output sequences are in general stochastic processes over the sample space of noise realizations.}
 \rev{For simplicity, the dynamical description~\eqref{eq:LTI_Network} assumes
  that all nodes are of order $1$, that is, $x(k+1)$ depends directly
  only on $x(k)$ and is conditionally independent of $\{x\}^{k-1}$
  given~$x(k)$. Nevertheless, as we discuss later (see e.g.,
  Remark~\ref{rem:higher-order}), all of the subsequent results are
  generalizable to systems whose dynamics (in the original
  ``physical'' variables) are described by difference equations of
  order higher than~$1$.}

Even though there is a control input at every node in the network
dynamics~\eqref{eq:LTI_Network}, we do not assume that all the control
inputs are user-specified.
In fact, in a large-scale network, it is common that one can \nrev{actuate} 
only a small subset of the nodes due to computational
constraints, physical limitations, or cost. A similar observation can
be made regarding the number of nodes whose state can be directly
\nrev{measured}. For these reasons, here we assume that the nodes of the
network are divided into $n_{m}\leq n$ \emph{manifest} nodes, which
can be directly \nrev{actuated and measured} by the user, and $n - n_m$
\emph{latent} nodes, which can neither be directly \nrev{actuated nor
measured} by the user. With this distinction, and using a permutation
of the indices in $(1,2,\ldots ,n)$ if necessary, we can decompose the
network and input state as $x = [x_m,x_l]$ and $u=[u_m,u_l]$,
respectively, where the subindex `$m$' corresponds to manifest nodes
and the subindex `$l$' corresponds to latent nodes.
With this convention, the \nrev{output matrix} takes the form $C=[I_{n_m
  \times n_m}, \mathbf{0}_{n_m \times (n-n_m)}]$.  With the
decomposition of the nodes into manifest and latent, the state-space
representation~\eqref{eq:LTI_Network} becomes
\begin{align}
  \begin{bmatrix}
    x_{m}(k+1)
    \\
    x_{l}(k+1)%
  \end{bmatrix}
  & =
  \begin{bmatrix}
    A_{11} & A_{12}
    \\
    A_{21} & A_{22}%
  \end{bmatrix}
  \begin{bmatrix}
    x_{m}(k)
    \\
    x_{l}(k)
  \end{bmatrix}
  +
  \begin{bmatrix}
    u_{m}(k)
    \\
    u_{l}(k)
  \end{bmatrix}
  , \notag
  \\
  y(k)& =x_{m}(k).  \label{eq:LTI_Network_partitioned}
\end{align}%
In the remainder of this paper, we consider the network in the
relabeled
form~\eqref{eq:LTI_Network_partitioned}. Fig.~\ref{fig:relabel}
illustrates this relabeling procedure \nrev{(corresponding to a
  linear transformation)} in a~ring.

\begin{figure}[tbh]
  \centering
  \includegraphics[scale=.65]{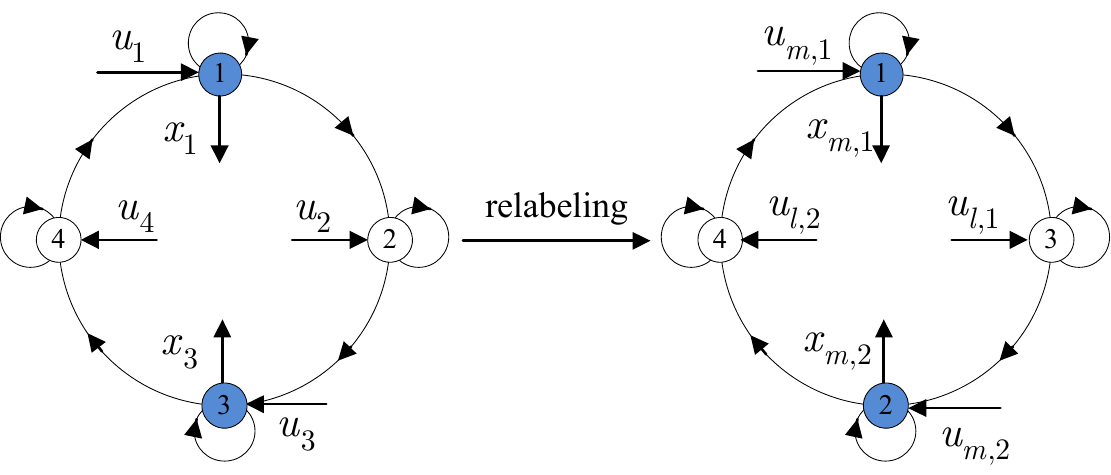}
  \caption{Node relabeling in a directed ring with $4$ nodes. Nodes
    $1$ and $3$ are manifest, nodes $2$ and $4$ are latent. The
    permutation $(1,2,3,4)\rightarrow (1,3,2,4)$ makes manifest and
    latent nodes have consecutive indices,
    as~in~\eqref{eq:LTI_Network_partitioned}.}\label{fig:relabel}
\end{figure}

\rev{
\nrev{Since the focus of this work is on network identification and not stabilization, we make the following standard assumption.}

\begin{assumption}\label{ass:assump1}
  The adjacency matrix of the complete network as well as the latent
  subnetwork are Schur stable, i.e., $\rho(A) < 1$ and $\rho(A_{22}) <
  1$.
\end{assumption}
}

\begin{remark}\longthmtitle{Direct versus latent interactions}
  The interaction graph of the manifest subnetwork is characterized by
  $A_{11}$.  In particular, the state of node $p$ affects the state of
  node $q$ \emph{directly} if and only if the entry on the $q$-th row
  and the $p$-th column, denoted by $A_{11}(q,p)$, is
  nonzero. However, even if $A_{11}(q,p)=0$, it is still possible that
  node $p$ affects node $q$ \emph{indirectly} through some latent
  nodes. The distinction between direct and indirect connections is an
  important point to which we come back later in our discussion.
  \oprocend
\end{remark}

We refer to a latent node as \emph{passive} if its corresponding input
is zero. Throughout the paper, we only deal with passive latent nodes,
so that $\{u_l\} \equiv 0$. We make the following assumption on the
input to the manifest nodes.

\begin{assumption}\label{ass:assump2}
  The input $\{u_{m}\}$ to the manifest subnetwork is a zero-mean
  stochastic process with independent and identically distributed
  (i.i.d.) \nnrev{absolutely continuous}%
  \footnote{\nnrev{Recall that an absolutely continuous random variable/vector is one that has a probability density function (e.g., Gaussian).}} 
  random vectors $u_{m}(k)$, with covariance $I_{n_{m}}$.
\end{assumption}

Assumption~\ref{ass:assump2} guarantees that 
$\{u_{m}\}$
is persistently exciting of arbitrary order and its power spectral
density does not vanish at any frequency.
 Similar
assumptions are common in system identification, see
e.g.,~\cite{LL:99,SS-VMP:15}. 
The zero-mean assumption can be
relaxed by assuming a nonzero but known $\E%
\lbrack u_{m}(k)]$ corresponding to the scenario where the
designer injects a deterministic stimulating signal into every
manifest node, which itself is subject to the disturbance of a
zero-mean white noise. Without loss of generality and for simplicity,
we assume $\E\lbrack u_{m}(k)]\equiv \mathbf{0}_{n_{m}}$.

Given the setup above, our objective is to identify the transfer
function $ T_{x_{m}u_{m}}(\omega )$ of the manifest subnetwork, that
is, the transfer function from $u_{m}$ to $x_{m}$, absent any
knowledge of the latent nodes. 

\begin{problem}\longthmtitle{Identification of the manifest transfer
    function}\label{Prob1}
  Given the measured data $\{y\}_{1}^{N}$, find a linear
  auto-regressive model of order $\tau $, with $N\gg \tau $, of the
  form
  \begin{equation}\label{NW_ML_CAN}
    \tilde{x}_{m}(k+1) = \sum_{i=0}^{\tau -1}\tilde{A}_{i}\tilde{x}
    _{m}(k-i)+u_{m}(k),
  \end{equation}
  such that the associated transfer function $T_{\tilde{x}_{m}u_{m}}$
  from $u_{m}$ to $\tilde{x}_{m}$ and the transfer function
  $T_{x_{m}u_{m}} $ from $u_{m}$ to $x_{m}$ in~\eqref{eq:LTI_Network}
  are close in the $H_{\infty }$-norm, i.e., $\lVert
  T_{\tilde{x}_{m}u_{m}}-T_{x_{m}u_{m}}\rVert _{\infty }$ is small.
\end{problem}

There are alternative methods to identify the transfer function matrix
$ T_{x_{m}u_{m}}$ besides the AR method in~\eqref{NW_ML_CAN}.
Our adoption here of AR model candidates is motivated by their
widespread use in neuroscience to determine causality and
interconnections in human brain connectivity models, see
e.g.,~\cite{AK-CMC-RK-PJF-NEC:08,SLB-AKS:11,JRI-AO-TM-MP-JS-GC-HP:14}%
\footnote{\nrev{In general, the main advantage of AR models over more general models such as
  ARMA or BJ is their simplicity, only capturing the internal dynamics and assuming negligible input \emph{noise correlation} (though putting no restriction on input \emph{signal correlation}, which is significant in brain dynamics). As a result, prediction error minimization has a closed-form solution for an AR model while it is non-convex in the ARMA or BJ cases.}}.
Equipped with time series data obtained during the performance of a
cognitive task,
 the \rev{conventional} procedure
consists of first estimating an AR model, then computing its
associated transfer function matrix, and finally evaluating the
Granger causality connectivity measure, or generalizations of it, in
the frequency domain. We are particularly motivated by the prospect of
understanding the sensitivity of these approaches to the presence of
latent nodes corresponding to brain regions that are active during the
cognitive task but are not directly \nrev{measured}.

\section{Asymptotically exact identification of the manifest transfer
  function}\label{Sec:AR_existence}

In this section we establish that, given an arbitrary precision, there
exists an AR model solving Problem~\ref{Prob1}. More precisely, we
show that there exists a sequence of AR models of the
form~\eqref{NW_ML_CAN} with increasing order whose transfer functions
converge to $T_{x_{m}u_{m}}$ exponentially in the $H_{\infty }$
sense. We later show that, if the latent subnetwork is acyclic,
then this approximation can be made exact.

\rev{
We start our discussion with a useful auxiliary result.

  \begin{lemma}\longthmtitle{Upper bound on
      $\|A_{22}^i\|$}\label{lem:rhobar}
    For any Schur stable $A_{22} \in \real^{n_l \times n_l}$ and any
    $\bar \rho \in (\rho(A_{22}), 1)$, there exists \nrev{$\kappa \in
      \realpos$} such that $ \|A_{22}^i\| \le \kappa \cdot \bar
    \rho^i$, for all $i \in \intnonneg$.
  \end{lemma}
  \begin{proof}
    The result is an immediate consequence of the spectral radius
    formula $\lim_{i \to \infty} \|A_{22}^i\|^{1/i} = \rho(A_{22})$.
  \end{proof}

We are now ready to state the main result of this section.

\begin{theorem}\longthmtitle{AR model whose transfer function
    converges to the manifest transfer function}\label{Theo1}
  Consider the LTI network described
  by~\eqref{eq:LTI_Network_partitioned} where all the latent nodes are
  passive.
   For any $\bar \rho \in (\rho(A_{22}), 1)$, there exists $\bar \gamma
  \in \realpos$ such that for all $\tau \in \intnonneg$, the AR
  model~\eqref{NW_ML_CAN} with 
  \begin{align}\label{eq:Pf_Theo_3}
    \tilde{A}_{0}^{\ast }=A_{11},\quad \tilde{A}_{i}^{\ast
    } = A_{12}A_{22}^{i-1}A_{21},  \ i \in \until{\tau -1},
  \end{align}
  guarantees
  \begin{align}\label{eq:Theo_1}
    \lVert T_{\tilde{x}_{m}u_{m}}(\cdot, \tau )-T_{x_{m}u_{m}}\rVert
    _{\infty } \le \bar{\gamma} \cdot \bar \rho^{\tau}.  
  \end{align}%
\end{theorem}
}
\begin{proof}
  We obtain from~\eqref{eq:LTI_Network_partitioned} that
  \begin{align}\label{eq:Pf_Theo_1}
    T_{x_{m}u_{m}}(\omega )&
    =(zI_{n_{m}}-A_{11}-A_{12}(zI_{n_{l}}-A_{22})^{-1}A_{21})^{-1}
    \notag
    \\
    & \overset{(a)}{=}(zI_{n_{m}}-A_{11}-\sum_{i=1}^{\infty
    }z^{-i}A_{12}A_{22}^{i-1}A_{21})^{-1}, 
  \end{align}%
  where $z=e^{j\omega }$ and $(a)$ follows by using the relation $%
  (zI_{n_{l}}-A_{22})^{-1}=\sum_{i=1}^{\infty
  }z^{-i}A_{22}^{i-1}$. Similarly, from~\eqref{NW_ML_CAN} we obtain
  \begin{equation}
    T_{\tilde{x}_{m}u_{m}}(\omega ,\tau )=(zI_{n_{m}}-\sum_{i=0}^{\tau
      -1}z^{-i}
    \tilde{A}_{i}^{\ast })^{-1}.  \label{eq:Pf_Theo_2}
  \end{equation}%
  Here we write the transfer function as
  $T_{\tilde{x}_{m}u_{m}}(\omega ,\tau ) $ to emphasize its dependence
  on $\tau $. It then follows directly that
  \begin{align*}
    & \lVert T_{\tilde{x}_{m}u_{m}}(\cdot, \tau
    )-T_{x_{m}u_{m}}\rVert _{\infty } \notag
    \\
    & =\lVert T_{x_{m}u_{m}}(T_{x_{m}u_{m}}^{-1} -T_{\tilde{x}%
      _{m}u_{m}}^{-1}(\cdot, \tau ))T_{\tilde{x}_{m}u_{m}}(\cdot, \tau )\rVert _{\infty } \notag
    \\
    & \overset{(a)}{\leq } \lVert T_{x_{m}u_{m}}\rVert
    _{\infty }\rVert T_{\tilde{x}_{m}u_{m}}(\cdot, \tau )\rVert
    _{\infty } \notag
     \lVert T_{x_{m}u_{m}}^{-1}-T_{\tilde{x}_{m}u_{m}}^{-1}(\cdot, \tau )\rVert _{\infty }
    \notag
    \\
    & \overset{(b)}{\leq }\lVert T_{x_{m}u_{m}}\rVert
    _{\infty }\rVert T_{\tilde{x}_{m}u_{m}}(\cdot, \tau )\rVert
    _{\infty }\sum_{i=\tau }^{\infty }\lVert
    z^{-i}A_{12}A_{22}^{i-1}A_{21}\rVert _{\infty } \notag
    \\
    & \rev{\overset{(c)}{\le} \lVert T_{x_{m}u_{m}}\rVert
    _{\infty }\rVert T_{\tilde{x}_{m}u_{m}}(\cdot, \tau )\rVert
    _{\infty } \|A_{12}\| \|A_{21}\| \sum_{i = \tau}^\infty \|A_{22}^{i - 1}\|} \notag
    \\
    & \rev{\overset{(d)}{\le} \gamma (\tau ) \cdot \bar \rho^\tau,}
  \end{align*}%
  where
    \begin{align}
      \gamma (\tau ) \triangleq \frac{\kappa \lVert
        T_{x_{m}u_{m}}\rVert_{\infty } \lVert A_{12}\rVert
        \rVert A_{21}\rVert}{\bar \rho - \bar \rho^2} \rVert
      T_{\tilde{x}_{m}u_{m}}(\cdot, \tau )\rVert _{\infty }
      . \label{eq:def_of_gamma}
    \end{align}
   Here, $(a)$ follows from the sub-multiplicativity of induced
  norms, $(b)$ follows by the sub-additivity of norms, $(c)$ follows
  by the definition of the $H_{\infty }$-norm and also the
  sub-multiplicativity of induced norms, \rev{and $(d)$ follows from
    Lemma~\ref{lem:rhobar}. The remainder of the proof is devoted to
    showing the existence of a uniform upper bound $\bar \gamma$ for
    $\gamma(\tau)$. By the definition of the $H_\infty$-norm,
    \begin{align}
      \lVert T_{\tilde{x}_{m}u_{m}}(\cdot, \tau )\rVert _{\infty }
      &= \sup_{-\pi \leq \omega \leq \pi } \sigma_{\max}\big(
      T_{\tilde{x}_{m}u_{m}}(\omega ,\tau)\big) \label{eq:Pf_Theo_5}
      \\
      &\overset{(a)}{=} \Big(\inf_{-\pi \leq \omega \leq \pi }\sigma
      _{\min}\big(%
      T_{\tilde{x}_{m}u_{m}}^{-1}(\omega ,\tau )\big)\Big)^{-1}, \notag
    \end{align}}
  where $(a)$ holds due to the fact that $\sigma _{\max}(M)=\sigma
  _{\min}^{-1}(M^{-1})$ for any invertible matrix $M$.
  \rev{ To complete the proof, we only need to show that
    \begin{align}\label{eq:infs}
      \vartheta \triangleq \inf_{\tau \in \intnonneg} \inf_{-\pi \leq
        \omega \leq \pi
      }\sigma_{\min}\big(T_{\tilde{x}_{m}u_{m}}^{-1}(\omega ,\tau )\big) >
      0.
    \end{align}
    We show this in two steps.
    \begin{enumerate}
    \item It follows from~\eqref{eq:Pf_Theo_1}
      and~\eqref{eq:Pf_Theo_2} that
      \begin{align*}
        \lim_{\tau \to \infty} T_{\tilde x_m u_m}^{-1}(\omega, \tau) =
        T_{x_m u_m}^{-1}(\omega), \quad \forall \omega \in [-\pi, \pi].
      \end{align*}
      It is straightforward to show, using the exponential decay of
      $A_{22}^\tau$ and \nrev{definition} of uniform convergence,
      that each entry of $T_{\tilde x_m u_m}^{-1}(\cdot, \tau)$ 
      converges uniformly to the corresponding
      entry of $T_{x_m u_m}^{-1}$. Hence, given the uniform
      continuity of matrix eigenvalues as a function of matrix
      entries~\cite[Thm 7.8c]{PH:88},
      $\sigma_{\min}\big(T_{\tilde{x}_{m}u_{m}}^{-1}(\cdot ,\tau
      )\big)$ converges uniformly to
      $\sigma_{\min}\big(T_{x_{m}u_{m}}^{-1}\big)$. Thus,
      since $\inf_{-\pi \le \omega \le \pi}
      \sigma_{\min}\big(T_{x_{m}u_{m}}^{-1}(\omega)\big)$ $= \|T_{x_m
        u_m}\|_\infty > 0$ (which itself holds by
      Assumption~\ref{ass:assump1}), there exists $\tau_0 \in
      \intnonneg$ such that
      \begin{align*}
        \inf_{\tau \ge \tau_0} \inf_{-\pi \leq \omega \leq \pi
        }\sigma_{\min}\big(T_{\tilde{x}_{m}u_{m}}^{-1}(\omega ,\tau )\big) >
        0.
      \end{align*} 
    \item For any finite $\tau$, we show that $T_{\tilde x_m u_m}(\cdot,
      \tau)$ is BIBO stable and thus has no poles on the unit circle
      (which in turn guarantees $\inf_{-\pi \leq \omega \leq \pi
      }\sigma_{\min}\big(T_{\tilde{x}_{m}u_{m}}^{-1}(\omega ,\tau
      )\big) > 0$). For any bounded input $u_m$, let the corresponding
      outputs of $T_{\tilde x_m u_m}(\cdot, \tau)$ and $T_{x_m u_m}$ be
      denoted by $\tilde x_m$ and $x_m$, resp. (with
      initial states set to zero). We then have 
      \begin{align*}
        x_m(k) - \tilde x_m(k) = A_{12} A_{22}^{\tau - 1} x_l(k - \tau),
      \end{align*}
      where $x_l$ is the (internal) state of the latent nodes in $T_{x_m
        u_m}$. By Assumption~\ref{ass:assump1}, both $x_m(k)$ and
      $A_{12} A_{22}^{\tau - 1} x_l(k - \tau)$ are bounded, proving the
      BIBO stability of $T_{\tilde x_m u_m}(\cdot, \tau)$.
    \end{enumerate}
    Hence, \eqref{eq:infs} follows by combining (i) and (ii) and the
    fact that the decomposition $\intnonneg = \{0\} \cup \{1\} \cup
    \cdots \cup \{\tau_0 - 1\} \cup \{\tau_0, \tau_0 + 1, \dots\}$ is
    finite. Equivalently, there exists $U > 0$ such that $\rVert
    T_{\tilde{x}_{m}u_{m}}(\cdot ,\tau )\rVert_\infty < U$ for all $\tau \in
    \intnonneg$, so~\eqref{eq:Theo_1} holds with $\bar \gamma = \kappa
    U \|T_{x_m u_m}\|_\infty \|A_{12}\| \|A_{21}\| / (\bar
    \rho - \bar \rho^2)$.  }
\end{proof}

Theorem~\ref{Theo1} shows that the presence of latent nodes in
the network, as long as they do not receive any external input, does
not affect the achievable accuracy of the identification via
AR modeling of the manifest transfer function.

\begin{remark}\longthmtitle{Direct versus latent
    interactions -- cont'd}\label{Remark_AR_t_domain}
  It follows from the network
  dynamics~\eqref{eq:LTI_Network_partitioned} that
  \begin{equation}
    x_{m}(k+1)=\sum_{i=0}^{k}\tilde{A}_{i}^{\ast
    }x_{m}(k-i)+A_{12}A_{22}^{k}x_{l}(0)+u_{m}(k).
    \label{eq1:Remark_AR_t_domain}
  \end{equation}%
  By virtue of \eqref{eq1:Remark_AR_t_domain}, we can distinguish
  whether two manifest nodes interact directly or indirectly through
  latent nodes by looking at the matrix sequence
  $\{\tilde{A}_{i}^{\ast }\}$. First, the state of manifest node $p$
  affects the state of manifest node $q$ directly if and only if
  $\tilde{A}_{0}^{\ast }(q,p)=A_{11}(q,p)\neq 0$.  Similarly, the
  state of manifest node $p$ affects the state of manifest node $q$
  indirectly through latent nodes if and only if $\tilde{A}_{i}^{\ast
  }(q,p)\neq 0$ for some $i\geq 1$. In particular, from the relation
  $\tilde{A}_{i}^{\ast }=-A_{12}A_{22}^{i-1}A_{21}$, one can see that
  the state of $p$ first affects some latent nodes (that correspond to
  the nonzero entries in the $p$-th column of~$A_{21}$) through
  $A_{21}$, then propagates through the latent subnetwork, reflected
  by $A_{22}^{i-1}$, and finally affects $q$ through~$A_{12}$.
  Furthermore, if the latent subnetwork is acyclic, then
  $\tilde{A}_{i}^{\ast }(q,p)\neq 0$ implies that there are exactly
  $i$ latent nodes in a path connecting $p$ to~$q$. \oprocend
\end{remark}

\rev{
  \begin{remark}\longthmtitle{Systems described by higher-order
      difference equations}\label{rem:higher-order}
    Unlike the system description in~\eqref{eq:LTI_Network}, the
    dynamic behavior of many real-world complex systems such as the
    brain cortical networks is described by difference equations of
    orders significantly greater than $1$, i.e.,
    \begin{align}\label{eq:LTI_Network-HO}
      x(k+1) &= A^{(0)} x(k) + A^{(1)} x(k - 1) + \cdots
      \\
      \notag &\quad + A^{(\nu - 1)} x(k - \nu + 1) +u(k), \qquad \nu
      \gg 1
    \end{align}
    where $x_1, \dots, x_{n_m}$ still denote the manifest (sensed and
    actuated) nodes and $x_{n_m + 1}, \ldots, x_n$ are the latent
    ones. In this description, the vector $x$ corresponds to some
    relevant physical variables.  Defining the state vector $\xi(k) =
    [x(k)^T \ x(k - 1)^T \ \cdots \ x(k - \nu + 1)^T]^T$, one can
    rewrite~\eqref{eq:LTI_Network-HO} in order-$1$
    form as
    \begin{align}\label{eq:LTI_Network-HO-o1}
      \left[\!\!\begin{array}{c} \xi_m(k + 1) \\ \xi_l(k +
          1) \end{array}\!\!\right] = \left[\!\!\begin{array}{cc}
          \Ac_{11} & \Ac_{12} \\ \Ac_{21} &
          \Ac_{22} \end{array}\!\!\right] \left[\!\!\begin{array}{c}
          \xi_m(k) \\ \xi_l(k) \end{array}\!\!\right] +
      \left[\!\!\begin{array}{c} u_m(k) \\ 0 \end{array}\!\!\right] ,
    \end{align}
    where $\xi_m(k) = x_m(k)$, $\xi_l(k) = [x_l(k)^T \ x_m(k - 1)^T \
    x_l(k - 1)^T \ \cdots \ x_m(k - \nu + 1)^T \ x_l(k - \nu +
    1)^T]^T$, $\Ac_{11} = A^{(0)}_{11}$, and
    \begin{align*}
      \Ac_{12} &= \left[\begin{array}{cccccc} A^{(0)}_{12} &
          A^{(1)}_{11} & A^{(1)}_{12} & \cdots & A^{(\tau - 1)}_{11} &
          A^{(\tau - 1)}_{12} \end{array}\right],
      \\
      \Ac_{21} &= \left[\begin{array}{cccccc} (A^{(0)}_{21})^T &
          I_{n_m} & 0 & \cdots & 0 & 0 \end{array}\right]^T,
      \\
      \Ac_{22} &= 
\setlength\arraycolsep{1pt}
 \begin{bmatrix}
 A^{(0)}_{22} & A^{(1)}_{21} &
        A^{(1)}_{22} & \cdots & A^{(\tau - 2)}_{21} & A^{(\tau - 2)}_{22} & A^{(\tau - 1)}_{21} & A^{(\tau - 1)}_{22}
        \\
        0 & 0 & 0 & \cdots & 0 & 0 & 0 & 0
        \\
        I_{n_l} & 0 & 0 & \cdots & 0 & 0 & 0 & 0
        \\
        \vdots & \vdots & \vdots & \ddots & \vdots & \vdots & \vdots & \vdots
        \\
        0 & 0 & 0 & \cdots & 0 & 0 & 0 & 0
        \\
        0 & 0 & 0 & \cdots & 0 & 0 & 0 & 0
        \\
        0 & 0 & 0 & \cdots & I_{n_m} & 0 & 0 & 0
        \\
        0 & 0 & 0 & \cdots & 0 & I_{n_l} & 0 & 0 \end{bmatrix} .
    \end{align*}
    In this description, we view $\xi_m$ as the actual ``manifest
    state'' of the system while the whole vector $\xi_l$ is the
    ``latent state''. The reason for this interpretation is that, at
    any time $k$, only $x_m(k)$ is directly sensed/actuated while $x(k
    - 1), \dots, x(k - \nu + 1)$ are quantities stored in the
    system. Interestingly, for the order-$1$
    description~\eqref{eq:LTI_Network}, this observation brings up the
    possibility of some of the latent variables $x_l$ simply being a
    relayed version of manifest variables. Note that, under this
    interpretation, the matrices $A^{(1)}_{11}, \dots, A^{(\nu -
      1)}_{11}$ represent manifest-latent (rather than
    manifest-manifest)
    interactions. From~\eqref{eq:LTI_Network-HO-o1}, it is clear that
    all the treatment for~\eqref{eq:LTI_Network} is readily
    applicable. Nevertheless, as $\nu$ increases, larger $\tau$ is
    necessary in order for~\eqref{NW_ML_CAN} to represent the system
    accurately. This is both intuitive and clear
    from~\eqref{eq:Theo_1} and~\eqref{eq:def_of_gamma}, where
    increasing $\nu$ results in larger $\|A_{12}\|$ and $\|A_{21}\|$
    as well as (usually) $\|T_{x_m u_m}\|$ and $\rho(A_{22})$. This,
    in turn, may result in numerical difficulties when one constructs
    the AR model from recorded input-output data (which is the subject
    of the next section). \oprocend
  \end{remark}
}

Next, we show that there exists an AR model~\eqref{NW_ML_CAN} whose
transfer function coincides with the manifest transfer function if the
latent subnetwork is acyclic.

\begin{corollary}\longthmtitle{Exact manifest transfer function
    identification for acyclic latent subnetworks}\label{Theo2}
  Under the \nnrev{assumptions} of Theorem~\ref{Theo1}, further assume that the
  latent subnetwork is acyclic, i.e., there exists $\tau _{22}\in
  \mathbb{Z}_{\geq 1}$ such that $A_{22}^{\tau
    _{22}}=\mathbf{0}_{n_{l}\times n_{l}}$. Then, the matrix sequence
  $\tilde{A}_{0}^{\ast },\cdots ,\tilde{A}%
  _{\tau _{22}}^{\ast }$ in~\eqref{eq:Pf_Theo_3}
  ensures~$T_{\tilde{x}_{m}u_{m}}=T_{x_{m}u_{m}}$.
\end{corollary}

The proof of the result follows by comparing~\eqref{eq:Pf_Theo_1}
and~\eqref{eq:Pf_Theo_2}, and using the \nnrev{assumption} that the latent
subnetwork is acyclic.  Theorem~\ref{Theo1} and Corollary~\ref{Theo2}
show that it is possible to identify the transfer function of the
manifest subnetwork without any knowledge of the passive latent
nodes. However,~\eqref{eq:Pf_Theo_3} cannot be directly applied to
determine the auto-regressive model because its evaluation requires
knowledge of the adjacency matrix $A$ of the whole network, which is
unknown. This problem can be circumvented by employing the measured
data sequence $\{y\}_{1}^{N} \subset \real^{n_m}$, as explained in the
next section.

\section{Identification via least-squares estimation}\label{Sec:AR_LS}

In this section we employ least-squares estimation to compute from
data the sequence of matrices defining the auto-regressive model. We
show that the estimates resulting from this method asymptotically
converge in probability, as the data length $N$ and model order $\tau
$ increase, to the optimal matrix sequence identified in
Theorem~\ref{Theo1}. Finally, we particularize our discussion to the
case of acyclic latent subnetworks.

\subsection{Least-squares auto-regressive estimation}

Given a vector sequence $\{y\}_{1}^{N} \subset \mathbb{R}^{n_{m}}$,
the problem of least-squares auto-regressive (LSAR) model estimation
with order $\tau \in \mathbb{Z}_{\geq 1}$ is to find a matrix sequence
$\{\hat{A}\}_{0}^{\tau -1} \subset \real^{n_m \times n_m}$ that
minimizes the $2$-norm of the residual sequence $\{e\}_{\tau }^{N-1}
\subset \real^{n_m}$ defined by
\begin{equation}\label{eq:AR}
  e(k) = y(k+1)-\sum_{i=0}^{\tau -1}\hat{A}_{i}y(k-i),
\end{equation}
for $k \in \{\tau, \dots, N-1\}$.  Equation~\eqref{eq:AR} can be
written in compact vector form as%
\begin{equation}\label{eq:AR_vec}
  \vec{y}_{N} = \hat{\mathbf{A}}_{\tau }\Phi
  _{N}+\vec{e}_{N},
\end{equation}
where
\begin{align*}
  \vec{y}_{N}& =
  \begin{bmatrix}
      y(\tau +1) & y(\tau +2) & \cdots & y(N)
  \end{bmatrix}
  \in \mathbb{R}^{n_{m}\times (N-\tau )},
  \\
  \vec{e}_{N} & =
  \begin{bmatrix}
    e(\tau ) & e(\tau +1) & \cdots & e(N-1)%
  \end{bmatrix}
  \in \mathbb{R}^{n_{m}\times (N-\tau )},
  \\
  \hat{\mathbf{A}}_{\tau }& =
  \begin{bmatrix}
    \hat{A}_{0} & \hat{A}_{1} & \cdots & \hat{A}_{\tau -1}%
  \end{bmatrix}
  \in \mathbb{R}^{n_{m}\times n_{m}\tau },
  \\
  \Phi _{N}& =
  \begin{bmatrix}
    y(\tau ) & y(\tau +1) & \cdots & y(N-1)
    \\
    y(\tau -1) & y(\tau ) & \cdots & y(N-2)
    \\
    \vdots & \vdots & \ddots & \vdots
    \\
    y(1) & y(2) & \cdots & y(N-\tau )%
  \end{bmatrix}.
\end{align*}%
\rev{Using the square of the prediction error~\cite{LL:99} 
  \begin{align*}
    \func{tr}(\vec e_N \vec e _N^T) = \func{tr}\big((\vec y_N -
    \hat{\mathbf A}_\tau \Phi_{N})(\vec y_N - \hat{\mathbf A}_\tau
    \Phi_{N})^T\big)
  \end{align*}
  as the cost function, we compute its gradient
  \begin{align*}
    \frac{\partial \func{tr}(\vec e_N \vec e _N^T)}{\partial
      \hat{\mathbf A}_\tau} = (\vec y_N - \hat{\mathbf A}_\tau \Phi_{N})
    (-\Phi_N^T) = \hat{\mathbf A}_\tau \Phi_N \Phi_N^T - \vec y_N
    \Phi_N^T.
  \end{align*} 
  Setting this to zero, we get a system of linear equations for which
  a solution is guaranteed to exist (since the rows of~$\vec y_N
  \Phi_N^T$ belong to the row space of $\Phi_N \Phi_N^T$, which is the
  same as the row space of
  $\Phi_N^T$).
  By Assumption~\ref{ass:assump2}, $\func{det}(\Phi_N \Phi_N^T) \neq
  0$ and this solution is unique with probability one.%
  \footnote{\nnrev{This is because (each element of) $\{y\}_{1}^{N-1}$ is an affine function of $\{u\}_{0}^{N-2}$, and $\func{det}(\Phi_N \Phi_N^T)$ is a polynomial function of $\{y\}_{1}^{N-1}$, so $\func{det}(\Phi_N \Phi_N^T)$ is a polynomial function of $\{u\}_{0}^{N-2}$. Therefore, the level set $\Nc = \setdef{\{u\}_{0}^{N-2}}{\func{det}(\Phi_N \Phi_N^T) = 0}$ has Lebesgue measure zero. Thus, by Assumption~\ref{ass:assump2}, $\Pr(\Nc) = 0$.}}
  If $\func{det}(\Phi_N \Phi_N^T) = 0$, the
  minimum-norm solution can be found as
  \begin{align}\label{eq:ls_AR}
    \hat{\mathbf A}_\tau = \vec y_N \Phi_N^T (\Phi_N \Phi_N^T)^{-1} =
    \vec y_N \Phi_N^+,
  \end{align}
  where $(\cdot)^+$ denotes the Moore-Penrose
  pseudo-inverse. Since~\eqref{eq:ls_AR} is also valid for the
  nonsingular case, it is taken as the solution to the LSAR estimation
  problem. In order to indicate the dependency of the solution upon
  the measured data sequence, we sometimes use the
  notation~$\hat{\mathbf{A}}_{\tau }(\{y\}_{1}^{N})$.  }

\subsection{Convergence in probability to manifest transfer function}

Here we study the transfer function resulting from the LSAR estimation
method and characterize its convergence properties, as the data length
and the model order increase, with respect to the transfer function of
the manifest subnetwork.  Our first result establishes that the LSAR
matrix estimate~\eqref{eq:ls_AR} converges in probability to the
optimal matrix sequence identified in Theorem~\ref{Theo1}.

\rev{
  \begin{proposition}\longthmtitle{The LSAR estimate converges in
      probability to optimal matrix sequence}\label{prop:LSAR-matrix}
    Consider the LTI network described
    by~\eqref{eq:LTI_Network_partitioned} where all latent nodes are
    passive. Given the measured data sequence $\{y\}_{1}^{N}$
    generated from the dynamics~\eqref{eq:LTI_Network_partitioned}
    stimulated by the white noise input $\{u_{m}\}$ according to
    Assumption~\ref{ass:assump2} and any $\bar \rho \in (\rho(A_{22}),
    1)$, there exists $\beta \in \realpos$ (depending only on the
    adjacency matrix $A$) such that the LSAR estimate
    $\hat{\mathbf{A}%
    }_{\tau }(\{y\}_{1}^{N})$ in~\eqref{eq:ls_AR} satisfies
    \begin{equation}
      \lVert \func{plim}_{N\rightarrow \infty }\hat{\mathbf{A}}_{\tau
      }(\{y\}_{1}^{N})-\tilde{\mathbf{A}}_{\tau }^{\ast }\rVert _{\max }\leq \beta
      \tau \bar \rho^\tau,  \label{eq:Lemma2}
    \end{equation}%
    where $ \tilde{\mathbf{A}}_{\tau }^{\ast } =
    \begin{bmatrix}
      \tilde{A}_{0}^{\ast } & \tilde{A}_{1}^{\ast } & \cdots &
      \tilde{A}_{\tau -1}^{\ast }
    \end{bmatrix}
    \in \mathbb{R}^{n_{m}\times n_{m}\tau }$ is the optimal matrix
    sequence given by~\eqref{eq:Pf_Theo_3}.
  \end{proposition}
}
\begin{proof}
\nnrev{For any quasi-stationary signal%
\footnote{Basically, a signal is quasi-stationary if it has a well-defined covariance function. See~\cite[Def 2.1]{LL:99} for a formal definition.}
$\{s\}$, let
\begin{align*}
R_s(j) \triangleq \lim_{N \to \infty} \frac{1}{N} \sum_{i = 1}^N \E[s(i) s(i - j)^T].
\end{align*} 
Using the Birkhoff's Ergodic Theorem~\cite[Thm 7.2.1]{RD:10} (see also~\cite[Thm 7.1.3]{RD:10}) and the fact that $\{y\}$ is the output of a stable system (and thus the effects of initial conditions asymptotically vanish), we can show that
\begin{align*}
    \func{plim}_{N\rightarrow \infty
    }\frac{1}{N}\sum_{i=1}^{N}y(i)y(i-j)^{T} = R_{y}(j).
  \end{align*}}
  As a result, $\frac{1}{N}\Phi _{N}\Phi _{N}^{T} \in
  \mathbb{R}^{n_{m}\tau \times n_{m}\tau }$ also converges in
  probability and
  \begin{align*}
    R_{\Phi }& \triangleq \func{plim}_{N\rightarrow \infty
    }\frac{1}{N}\Phi _{N}\Phi _{N}^{T}
    \\
    & =
    \begin{bmatrix}
      R_{y}(0) & R_{y}(1) & \cdots & R_{y}(\tau -1)
      \\
      R_{y}^{T}(1) & R_{y}(0) & \cdots & R_{y}(\tau -2)
      \\
      \vdots & \vdots & \ddots & \vdots
      \\
      R_{y}^{T}(\tau -1) & R_{y}^{T}(\tau -2) & \cdots & R_{y}(0)%
    \end{bmatrix}
    .
  \end{align*}%
  Define
  \begin{equation}
    \nu (k)\triangleq y(k+1)-\sum_{i=0}^{\tau -1}\tilde{A}_{i}^{\ast }y(k-i),
    \label{Def_nu}
  \end{equation}%
  and note that the transfer function from $u_{m}$ to $\nu $ is
  $T_{\tilde{x}_{m}u_{m}}^{-1}T_{x_{m}u_{m}}$, where $T_{x_{m}u_{m}}$
  and $T_{\tilde{x}_{m}u_{m}}$ are given by~\eqref{eq:Pf_Theo_1} and~%
  \eqref{eq:Pf_Theo_2}, respectively. Equation~\eqref{Def_nu} can be
  written in compact vector form as
  \begin{equation}
    \vec{y}_{N}=\tilde{\mathbf{A}}_{\tau }^{\ast }\Phi _{N}+\vec{\nu}_{N},
    \label{eq1:Pf_Theo_3}
  \end{equation}%
  with $ \vec{\nu}_{N}\triangleq [ \nu (\tau ) \hspace*{1ex}  \nu (\tau +1) \hspace*{1ex} 
  \cdots \hspace*{1ex}  \nu (N-1)] \in \mathbb{R}^{n_{m}\times (N-\tau )}$.
  From~\eqref{eq:ls_AR} and~\eqref{eq1:Pf_Theo_3}, it follows that
  \begin{align}
    \func{plim}_{N\rightarrow \infty }\hat{\mathbf{A}}_{\tau
    }(\{y\}_{1}^{N})& =%
    \func{plim}_{N\rightarrow \infty }\frac{1}{N}\vec{y}_{N}\Phi
    _{N}^{T}(\frac{1%
    }{N}\Phi _{N}\Phi _{N}^{T})^{-1} \notag
    \\
    & \!\!\!\!\!\!\!\!=\tilde{\mathbf{A}}_{\tau }^{\ast }+\func{plim}_{N\rightarrow
      \infty }\frac{1%
    }{N}\vec{\nu}_{N}\Phi _{N}^{T} (\frac{1%
    }{N}\Phi _{N}\Phi _{N}^{T})^{-1}.  \label{eq2:Pf_Lemma2}
  \end{align}%
  Moreover, Assumption~\ref{ass:assump2} renders $u_{m}(k)$
  independent of $%
  \{y\}_{1}^{k}$, which further implies that $
  \func{plim}_{N\rightarrow \infty }\frac{1}{N}\vec{u}_{m,N}\Phi
  _{N}^{T}=%
  \mathbf{0}_{n_{m}\times n_{m}\tau }$, where $\vec{u}_{m,N}\triangleq
  [ u_{m}(\tau ) \hspace*{1ex} u_{m}(\tau +1) \hspace*{1ex} \cdots
  \hspace*{1ex} u_{m}(N-1) \in \real^{n_m \times (N-\tau)} $.
  Therefore,
  \begin{equation}
    \func{plim}_{N\rightarrow \infty }\frac{1}{N}\vec{\nu}_{N}\Phi _{N}^{T}=%
    \func{plim}_{N\rightarrow \infty }\frac{1}{N}(\vec{\nu}_{N}-\vec{u}%
    _{m,N})\Phi _{N}^{T}=\Psi ,  \label{eq3:Pf_Lemma2}
  \end{equation}%
  where $\Psi \triangleq
  \begin{bmatrix}
    \Psi _{1} & \Psi _{2} & \cdots & \Psi _{\tau }
  \end{bmatrix} \in \real^{n_m \times n_m \tau} $, with
  \begin{equation*}
    \Psi _{j}\triangleq \func{plim}_{N\rightarrow \infty }\frac{1}{N}%
    \sum_{i=\tau }^{N-1}(\nu (i)-u_{m}(i))y^{T}(i-j+1) \in \real^{n_m \times n_m} .
  \end{equation*}%
  Thus, using $\func{plim}_{N \to \infty} (\frac{1%
    }{N}\Phi _{N}\Phi _{N}^{T})^{-1} = R_\Phi^{-1}$, we have
    \begin{align*}
    \func{plim}_{N\rightarrow \infty }\hat{\mathbf{A}}_{\tau
    }(\{y\}_{1}^{N}) - \tilde{\mathbf{A}}_{\tau }^{\ast } = \Psi R_\Phi^{-1}.
    \end{align*}
  By the sub-additivity of the max norm, it holds for any $j \in
  \until{\tau}$ that
  \begin{align}\label{eq4:Pf_Lemma2}
    \norm{\Psi _{j}}_{\max} & \leq \func{plim}_{N\rightarrow \infty
    }\frac{1}{N}\sum_{i=\tau }^{N-1}\lVert (\nu
    (i)-u_{m}(i))y^{T}(i-j+1)\rVert _{\max } \notag
    \\
    & \rev{\overset{(a)}{\leq }\func{plim}_{N\rightarrow \infty }
    \frac{\bar \rho^{-\tau}}{N}\sum_{i=\tau }^{N-1}(\nu
    (i) \!-\!u_{m}(i)\hspace*{-1pt})^{T}(\nu (i)
    \!-\!u_{m}(i)\hspace*{-1pt}) } \notag
    \\
    & \rev{\quad +\func{plim}_{N\rightarrow \infty }\frac{\bar \rho^\tau}{N}\sum_{i=\tau
    }^{N-1}y^{T}(i-j+1)y(i-j+1) } \notag
    \\
    & \rev{= \hspace*{-1pt} \bar \rho^{-\tau} \hspace*{-2pt}
    \func{tr}(R_{\nu-u_{m}}(0)\hspace*{-1pt})
    \hspace*{-1pt}+\hspace*{-1pt} \bar \rho^\tau \hspace*{-2pt}
    \func{tr}(R_{y}(0)\hspace*{-1pt})} ,
  \end{align}
  where $(a)$ follows from Lemma~\ref{Lemma1} in the appendix with the
  positive scalar $M$ chosen as \rev{$\bar \rho^\tau$}.  Using
  the fact that the transfer function from $u_{m}$ to $\nu-u_{m}$ is
  $T_{\tilde{x}_{m}u_{m}}^{-1}T_{x_{m}u_{m}}-I_{n_{m}}$, we obtain
  \begin{align}
    R_{\nu-u_{m}}(0) &\triangleq \lim_{N\rightarrow
      \infty }\frac{1}{N}\sum_{i=0}^{N
      -1}\E\lbrack (\nu-u_{m})(i)(\nu-u_{m})^{T}(i)]
      \notag
    \\
    & \overset{(a)} =\frac{1}{2\pi }\int_{-\pi }^{\pi }(T_{%
      \tilde{x}_{m}u_{m}}^{-1}T_{x_{m}u_{m}}(\omega )-I_{n_{m}})
    \notag
    \\
    & \quad \cdot (T_{\tilde{x}_{m}u_{m}}^{-1}T_{x_{m}u_{m}}(\omega
    )-I_{n_{m}})^{\ast }d\omega \notag
    \\
    & \overset{(b)}{\leq }\lVert
    T_{\tilde{x}_{m}u_{m}}^{-1}T_{x_{m}u_{m}}-I_{n_{m}}\rVert _{\infty
    }^{2}I_{n_{m}} \notag
    \\
    & \overset{(c)}{\leq }\lVert T_{x_{m}u_{m}}-T_{\tilde{x}%
      _{m}u_{m}}\rVert _{\infty }^{2}\lVert
    T_{\tilde{x}_{m}u_{m}}^{-1}\rVert _{\infty }^{2}I_{n_{m}} \notag
    \\
    & \overset{(d)}{\leq }\hat{\gamma} \rev{\bar \rho^{2 \tau}} I_{n_{m}},
    \label{eq5:Pf_Lemma2}
  \end{align}%
  where \rev{$ \hat{\gamma} \triangleq \bar{\gamma}^{2}\big(1+\lVert
    A_{11}\rVert +\lVert A_{12}\rVert \lVert A_{21}\rVert \kappa
    (1-\bar \rho)^{-1} \big)^{2}$ is constant}, $(a)$ follows
  from~\cite[eq.~(9-193)]{AP-SUP:02}, $(b)$ follows by the definition
  of $H_{\infty }$-norm, $(c)$ follows by the sub-multiplicativity of
  induced norms, and $(d)$ holds because of Theorem~\ref{Theo1} and
  the observation that, \rev{by Lemma~\ref{lem:rhobar},
    \begin{align*}
      \lVert T_{\tilde{x}_{m}u_{m}}^{-1}\rVert _{\infty }\leq 1+\lVert
      A_{11}\rVert +\lVert A_{12}\rVert \lVert A_{21}\rVert \kappa
      (1-\bar \rho)^{-1}.
    \end{align*}} 
  We obtain from~\eqref{eq4:Pf_Lemma2}
  and~\eqref{eq5:Pf_Lemma2},
  \begin{equation*}
    \lVert \Psi _{j}\rVert _{\max }\leq \rev{\bar \rho^\tau} (\hat{\gamma}%
    n_{m}+\func{tr}(R_{y}(0))),
  \end{equation*}%
  and from~\eqref{eq2:Pf_Lemma2} and~\eqref{eq3:Pf_Lemma2},
  \begin{align*}
    & \rVert \func{plim}_{N\rightarrow \infty }\hat{\mathbf{A}}_{\tau
    }(\{y\}_{1}^{N})-\tilde{\mathbf{A}}_{\tau }^{\ast }\rVert _{\max
    }=\rVert \Psi R_{\Phi }^{-1}\rVert _{\max }
    \\
    & \quad \leq n_{m}\tau \rVert R_{\Phi }^{-1}\rVert _{\max }\rVert
    \Psi \rVert _{\max }
    \\
    & \quad =n_{m}\tau \rVert R_{\Phi }^{-1}\rVert _{\max
    }\max_{j}\rVert \Psi _{j}\rVert _{\max } \leq \beta \tau \rev{\bar \rho^\tau},
  \end{align*}%
  where $\beta =(\hat{\gamma}n_{m}^{2}+\func{tr}(R_{y}(0))n_{m})\rVert
  R_{\Phi }^{-1}\rVert _{\max }$, as claimed.
\end{proof}

When it is clear from context, we refer to $\func{plim}_{N\rightarrow
  \infty }\hat{A}_{i}(\{y\}_{1}^{N})$ simply as $\hat{A}_{i}$.

\begin{remark}\longthmtitle{Direct versus latent interactions --
    cont'd}\label{re:LSAR-direct-latent}
  Proposition~\ref{prop:LSAR-matrix} shows that $\hat{A}_{i}$
  converges in probability to $\tilde{A}_{i}^{\ast }$ exponentially as
  the model order $\tau $ increases. Therefore, within a margin of
  error that can be tuned as desired, we deduce from the discussion in
  Remark~\ref{Remark_AR_t_domain} that the LSAR estimate~$\hat{A}_{0}$
  allows us to determine whether two manifest nodes interact directly
  and the LSAR estimates $\{\hat{A}_{i}\}_{i \geq 1}$ allow us to
  determine whether two manifest nodes interact indirectly through
  latent nodes with high probability as the length of measurement data
  grows. \oprocend
\end{remark}

Given the result in Proposition~\ref{prop:LSAR-matrix}, we next turn
our attention to the transfer function from $e$ to $y$ resulting from
the LSAR estimation~\eqref{eq:AR}, which we denote by
$T_{ye}(\{y\}_{1}^{N},\tau )$.  The next result shows that the
$H_{\infty}$-norm of this transfer function is uniformly upper bounded
with respect to the model order~$\tau$.

\begin{lemma}\longthmtitle{$H_{\infty}$-norm of $T_{ye}$ is uniformly
    upper bounded}\label{Lemma3}
  Under the \nnrev{assumptions} of Proposition~\ref{prop:LSAR-matrix}, there
  exist positive scalars $\tau _{0}$ and $ U_{T_{ye}}^{\infty }$ such
  that, for $\tau \geq \tau _{0}$,
  \begin{equation}
    \lVert \func{plim}_{N\rightarrow \infty }T_{ye}(\{y\}_{1}^{N},\tau )\rVert
    _{\infty }\leq U_{T_{ye}}^{\infty }.  \label{eq:Lemma3}
  \end{equation}
\end{lemma}
\begin{proof}
  By definition of $H_{\infty }$-norm, we have
  \rev{
  \begin{multline}
    \lVert \func{plim}_{N\rightarrow \infty }T_{ye}(\{y\}_{1}^{N},\tau
    )\rVert _{\infty } = \sup_{-\pi \leq \omega \leq \pi }\sigma_{\max}
    \big(\func{plim}%
    _{N\rightarrow \infty }T_{ye}(\omega ,\tau )\big)
    \\
    =\Big(\inf_{-\pi \leq \omega \leq \pi }\sigma_{\min}
    \bigl(\func{plim}_{N\rightarrow \infty }T_{ye}^{-1}(\omega ,\tau
    )\bigr)\Big)^{-1} .
    \label{eq1:Pf_Lemma3}
  \end{multline}
  Note that, for every $\omega \in [-\pi, \pi]$ and $\tau \in \intnonneg$,
  \begin{align}\label{eq:hat-tilde}
    \notag \func{plim}_{N \to \infty}T_{ye}^{-1}(\omega, \tau) &= z
    I_{n_m} - \sum_{i = 0}^{\tau - 1} z^{-i} \hat A_i
    \\
    &= T_{\tilde x_m u_m}^{-1} (\omega, \tau) - \sum_{i = 0}^{\tau -
      1} z^{-i} (\hat A_i - \tilde A_i^*),
  \end{align} 
  where $z = e^{j \omega}$. However, for every $\omega \in [-\pi,
  \pi]$ and $\tau \in \intnonneg$,
  \begin{align*}
    \! \Big\|\sum_{i=0}^{\tau -1} z^{-i} (\hat{A}_{i} - \tilde{A}_{i}^*)\Big\|
    & \! \le \! \sum_{i = 0}^{\tau - 1} \|\hat A_i - \tilde A_i^*\| 
    \! \stackrel{(a)}{\leq}\! n_{m}\sum_{i=0}^{\tau -1}\rVert
    \hat{A}_{i}-\tilde{A}_{i}^{\ast }\rVert _{\max } \notag
    \\
    \!  & \!\stackrel{(b)}{\leq} n_{m}\tau \max_{i}\rVert
    \hat{A}_{i}-\tilde{A}_{i}^{\ast }\rVert _{\max } \leq n_{m}\beta
    \tau ^{2} \bar \rho^\tau,
  \end{align*}%
  where $(a)$ follows from the fact that $\|A\| \le n_m \|A\|_{\max}$
  for any matrix $A \in \real^{n_m \times n_m}$ and $(b)$ follows from
  Proposition~\ref{prop:LSAR-matrix}.  Therefore, using Weyl's theorem
  for the perturbation of singular values~\cite{HW:12}
  in~\eqref{eq:hat-tilde} and taking $\inf_{-\pi \le \omega \le \pi}$
  of both sides, we get
  \begin{align*}
    &\inf_{-\pi \le \omega \le \pi} \sigma_{\min} \big(\func{plim}_{N
      \to \infty} T_{ye}^{-1}(\omega ,\tau)\big)
    \\
    &\qquad \ge \inf_{-\pi \le \omega \le \pi}
    \sigma_{\min}\big(T_{\tilde x_m u_m}^{-1}(\omega, \tau)\big) -
    \Big\|\sum_{i=0}^{\tau -1} z^{-i} (\hat{A}_{i} -
    \tilde{A}_{i}^*)\Big\|
    \\
    &\qquad \ge \inf_{-\pi \le \omega \le \pi}
    \sigma_{\min}\big(T_{\tilde x_m u_m}^{-1}(\omega, \tau)\big) -
    n_{m}\beta \tau ^{2} \bar \rho^\tau.
    \end{align*} 
    In view of~\eqref{eq:infs}, let $\tau_0$ be such that
    \begin{align}\label{eq5:Pf_Lemma3}
      n_{m}\beta \tau ^{2} \bar \rho^\tau \le \frac{\vartheta}{2},
      \qquad \forall \tau \ge \tau_0.
    \end{align}
    Then, the result follows from~\eqref{eq1:Pf_Lemma3} with
    $U_{T_{ye}}^\infty = \frac{2}{\vartheta}$.
  }
\end{proof}

We are finally ready to show that the transfer function $T_{ye}$
obtained from the LSAR method converges in probability to the transfer
function~$T_{x_{m}u_{m}}$ of the manifest subnetwork.

\begin{theorem}\longthmtitle{The LSAR method consistently estimates
    the manifest transfer function}\label{Theo3}
  Under the \nnrev{assumptions} of Proposition~\ref{prop:LSAR-matrix}, \rev{for
    any $\bar \rho \in (\rho(A_{22}), 1)$,} there exist positive
  scalars $\bar{\beta}$, $\bar{\gamma}$ and $ \tau _{0}$ such that,
  for $\tau \geq \tau _{0}$,
  \begin{equation}
    \lVert \func{plim}_{N\rightarrow \infty }T_{ye}(\{y\}_{1}^{N},\tau
    )-T_{x_{m}u_{m}}\rVert
    _{\infty }\leq (\bar{\beta}\tau ^{2}+\bar{\gamma}) \rev{\bar \rho^\tau}.  \label{eq2:Theo3}
  \end{equation}
  Consequently, $\func{plim}_{N\rightarrow \infty ,\tau \rightarrow
    \infty }T_{ye}(\{y\}_{1}^{N},\tau )=T_{x_{m}u_{m}}$.
\end{theorem}
\begin{proof}
  We only need to prove~\eqref{eq2:Theo3} as it directly implies the last equation in the
  statement. By the sub-additivity and
  sub-multiplicity of induced norms,
\begin{align}
    & \lVert T_{ye}(\cdot, \tau )-T_{x_{m}u_{m}}\rVert _{\infty } \notag
    \\
    &  \leq \lVert T_{ye}(\cdot, \tau )-T_{\tilde{x}_{m}u_{m}}(\cdot, \tau
    )\rVert _{\infty }+\lVert T_{\tilde{x}_{m}u_{m}}(\cdot, \tau
    )-T_{x_{m}u_{m}}\rVert _{\infty } \notag
    \\
    &  \leq \norm{T_{ye}(\cdot, \tau )}_{\infty }
    \norm{T_{\tilde{x}_{m}u_{m}}(\cdot, \tau )}_{\infty}
    \norm{T_{ye}^{-1}(\cdot, \tau )-T_{\tilde{x}_{m}u_{m}}^{-1}(\cdot, \tau
      )}_{\infty } \notag
    \\
    &  \quad +\lVert T_{\tilde{x}_{m}u_{m}}(\cdot, \tau
    )-T_{x_{m}u_{m}}\rVert _{\infty }.  \label{eq1:Pf_Theo3}
  \end{align}%
  Next, by~\eqref{eq:Pf_Theo_5}, Lemma \ref{Lemma3}, and
  Theorem~\ref{Theo1}, there exist positive scalars $\tau _{0}$,
  $U_{T_{ye}}^{\infty }$ and \rev{$\vartheta$} such that for $\tau \geq \tau
  _{0}$,
  \begin{align}
    & \lVert \func{plim}_{N\rightarrow \infty }T_{ye}(\cdot, \tau
    )-T_{x_{m}u_{m}}\rVert _{\infty } \label{eq2:Pf_Theo3}
    \\
    & \quad \leq U_{T_{ye}}^{\infty
    } \rev{\vartheta^{-1}} \lVert
    \func{plim}_{N\rightarrow \infty }T_{ye}^{-1}(\cdot, \tau )-T_{\tilde{x}%
      _{m}u_{m}}^{-1}(\cdot, \tau )\rVert _{\infty }+\bar{\gamma} \rev{\bar \rho^\tau}.  \notag
  \end{align}%
  Finally, according to the definition of $T_{ye}(\cdot, \tau )$
  in~\eqref{eq:AR} and $T_{\tilde{x}_{m}u_{m}}(\cdot, \tau )$
  in~\eqref{eq:Pf_Theo_2}, it follows that
  \begin{align}
    & \lVert \! \func{plim}_{N\rightarrow \infty } \! T_{ye}^{-1}(\cdot, \tau
    )\!-\!T_{\tilde{x}%
      _{m}u_{m}}^{-1}(\cdot, \tau )\rVert _{\infty }\!=\!\lVert \sum_{i=0}^{\tau
      -1} \! z^{-i}(%
    \!\func{plim}_{N\rightarrow \infty } \! \hat{A}_{i}\!-\!\tilde{A}_{i}^{\ast
    })\rVert _{\infty } \notag
    \\
    & \quad \overset{(a)}{\leq }\sum_{i=0}^{\tau -1}\lVert
    \func{plim}%
    _{N\rightarrow \infty }\hat{A}_{i}-\tilde{A}_{i}^{\ast }\rVert
    \overset{(b)}{%
      \leq }n_{m}\beta \tau ^{2} \rev{\bar \rho^\tau}, \label{eq3:Pf_Theo3}
  \end{align}%
  where $(a)$ holds by the sub-additivity and sub-multiplicity of
  $\|\cdot\|$ and $(b)$ follows by
  Proposition~\ref{prop:LSAR-matrix} and the fact that $\rVert A\rVert
  \leq n_{m}\rVert A\rVert _{\max }$ for any matrix $A\in
  \mathbb{R}^{n_{m}\times n_{m}}$. Thus, we
  obtain~\eqref{eq2:Theo3} for $ \tau \geq \tau _{0}$, where
  \rev{$\bar{\beta}\triangleq U_{T_{ye}}^{\infty } \vartheta^{-1} n_{m}\beta$} is a constant.
\end{proof}

According to Theorem~\ref{Theo3}, when the length $N$ of the
measurement data is sufficiently large
and the model order $\tau$ exceeds a certain threshold, the error
$\lVert T_{ye}(\tau) - T_{x_{m}u_{m}} \rVert_{\infty }$ obtained by
the LSAR method decreases exponentially with~$\tau$.

\begin{remark}\longthmtitle{Identification of
    manifest transfer function requires higher-order models as stability
    margin of latent subnetwork decreases}
  Even though an explicit expression of the threshold $\tau_0$ in
  Theorem~\ref{Theo3} as a function of the network is difficult to
  obtain, we can still make some useful observations.  \rev{From
    inequality~\eqref{eq5:Pf_Lemma3} in the proof of
    Lemma~\ref{Lemma3}, one can see that
    $\tau _{0}$ is an increasing function of~$\bar \rho$. Hence,
    as the latent subnetwork becomes less stable ($\rho (A_{22})$
    gets closer to $1$),} the corresponding $\tau _{0}$ becomes larger,
  requiring the order of the AR model to be higher to ensure exponential
  convergence. \oprocend
\end{remark}

\rev{
  \begin{remark}\longthmtitle{Systems described by higher-order
      difference equations -- cont'd}\label{rem:higher-order2}
    As explained in Remark~\ref{rem:higher-order}, the AR
    representation of systems with order $\nu > 1$ is identical to the
    $\nu = 1$ case, although they require larger AR order~$\tau$. For
    large-scale systems ($n \gg 1$), increasing $\tau$ rapidly raises
    the number of parameters in~\eqref{eq:AR_vec}, which leads to
    over-parametrization of the LSAR identification. Our simulations
    in Section~\ref{sec:simulations} show how this can be overcome
    both by increasing $N$ (which is computationally costly) and
    exponential regularization.  Also, note that when $\nu > 1$, the
    only member of the sequence of matrices $A^{(0)}_{11}, \dots,
    A^{(\nu - 1)}_{11}$ (denoting all current and past interactions
    among manifest \emph{nodes}) that is identifiable by the LSAR
    method is $A^{(0)}_{11}$ (representing direct interactions among
    manifest \emph{states}) while the others are only identifiable in
    the aggregate form~\eqref{eq:Theo_1}.  \oprocend
  \end{remark}
}

\subsection{Exact identification for acyclic latent subnetworks}

Here we show that the transfer function of the manifest
subnetwork can be perfectly identified using the LSAR method with a
finite model order if the latent subnetwork is acyclic. We start by
refining the result in Proposition~\ref{prop:LSAR-matrix} and showing
how, in this case, the convergence of the LSAR matrix
estimate~\eqref{eq:ls_AR} to the optimal matrix sequence identified in
Theorem~\ref{Theo1} holds in the mean-square sense.

\begin{proposition}\longthmtitle{The LSAR estimate converges in mean
    square to optimal matrix sequence for acyclic  latent
    subnetworks}\label{prop:LSAR-mean-square}
  Consider the LTI network described
  by~\eqref{eq:LTI_Network_partitioned} where all latent nodes
  are passive.  Further assume that the
  latent subnetwork is acyclic, i.e., there exists $\tau _{22}\in
  \mathbb{Z}_{\geq 1}$ such that $A_{22}^{\tau
    _{22}}=\mathbf{0}_{n_{l}\times n_{l}}$. Given the measured data
  sequence $\{y\}_{1}^{N}$ generated from the
  dynamics~\eqref{eq:LTI_Network_partitioned} stimulated by the white
  noise input $\{u_{m}\}$ according to Assumption~\ref{ass:assump2},
  the LSAR estimate $\hat{\mathbf{A}}_{\tau }(\{y\}_{1}^{N})$
  in~\eqref{eq:ls_AR} satisfies, for any $\tau \geq \tau _{22}+1$,
  \begin{equation*}
    \lim_{N\rightarrow \infty }\E\lbrack (\hat{\mathbf{A}}_{\tau
    }(\{y\}_{1}^{N})-%
    \tilde{\mathbf{A}}_{\tau }^{\ast })^{T}(\hat{\mathbf{A}}_{\tau
    }(\{y\}_{1}^{N})-%
    \tilde{\mathbf{A}}_{\tau }^{\ast })]=\mathbf{0}_{n_{m}\tau \times n_{m}\tau }.
  \end{equation*}
\end{proposition}
\begin{proof}
  If $A_{22}$ is nilpotent, using Corollary~\ref{Theo2}, we deduce
  that the transfer function from $u_{m}$ to $\nu $ defined
  in~\eqref{Def_nu} is $T_{\tilde{x}%
    _{m}u_{m}}^{-1}T_{x_{m}u_{m}}=I_{n_{m}}$. Consequently, the random
  vectors $\nu (k)$'s are i.i.d. with zero mean and finite second
  moment $\E\lbrack \nu (k)\nu ^{T}(k)] = I_{n_{m}}$.  Define
  \begin{align*}
    Z_{N}& \triangleq \frac{1}{N}(\hat{\mathbf{A}}_{\tau
    }-\tilde{\mathbf{A}}_{\tau }^{\ast })\Phi _{N}\Phi _{N}^{T}
    \\
    & \overset{(a)}{=}\frac{1}{N}(\vec{\nu}_{N}-\vec{e}_{N})\Phi
    _{N}^{T} \overset{(b)}{=}\frac{1}{N}\vec{\nu}_{N}\Phi _{N}^{T},
  \end{align*}
  where $(a)$ follows from~\eqref{eq:AR_vec} and~\eqref{eq1:Pf_Theo_3}
  and $(b) $ follows from the fact that the least-squares estimate
  $\hat{\mathbf{A}}_{\tau }$ in~\eqref{eq:ls_AR} renders
  $\vec{e}_{N}\Phi _{N}^{T} = \mathbf{0}%
  _{n_{m}\times n_{m}\tau }$. Combining the fact that the $\nu (k)$'s
  are i.i.d. and the fact that $\{y\}_{1}^{k}$ is a function of
  $\{\nu\}_{1}^{k-1}$, we deduce that $\nu (k)$ are independent of
  $\{y\}_{1}^{k}$.  This further implies that $\E%
  \lbrack Z_{N}]=\mathbf{0}_{n_{m}\times n_{m}\tau }$. Furthermore,
  \begin{align*}
    \lim_{N\rightarrow \infty }\E\lbrack Z_{N}^{T}Z_{N}]&
    =\lim_{N\rightarrow \infty }\frac{1}{N^{2}}\E\lbrack \Phi
    _{N}\vec{\nu}_{N}^{T}\vec{\nu}_{N}\Phi _{N}^{T}]
    \\
    & =\lim_{N\rightarrow \infty }\frac{1}{N}R_{\Phi }
    =\mathbf{0}_{n_{m}\tau \times n_{m}\tau }.
  \end{align*}
  Therefore, $ \lim_{N\rightarrow \infty }\E\lbrack
  \hat{\mathbf{A}}_{\tau }-\tilde{\mathbf{A}}%
  _{\tau }^{\ast }] =\lim_{N\rightarrow \infty }\E\lbrack
  Z_{N}]R_{\Phi }^{-1} =\mathbf{0}_{n_{m}\times n_{m}\tau }$ and $
  \lim_{N\rightarrow \infty }\E\lbrack (\hat{\mathbf{A}}_{\tau
  }-\tilde{\mathbf{A}}%
  _{\tau }^{\ast })^{T}(\hat{\mathbf{A}}_{\tau
  }-\tilde{\mathbf{A}}_{\tau }^{\ast })] =R_{\Phi
  }^{-1}\lim_{N\rightarrow \infty }\E\lbrack Z_{N}^{T}Z_{N}]R_{\Phi
  }^{-1} =\mathbf{0}_{n_{m}\tau \times n_{m}\tau }$,
   as claimed.
\end{proof}

We build on this result to show that the manifest transfer function
can be perfectly identified using the LSAR method with a finite model
order if the latent subnetwork is acyclic.

\begin{theorem}\longthmtitle{Exact manifest transfer function
    identification for acyclic latent subnetworks}\label{Theo3_nil}
  Under the \nnrev{assumptions} of Proposition~\ref{prop:LSAR-mean-square}, for
  any $\tau \geq \tau _{22}+1$,
  \begin{equation*}
    \func{plim}_{N\rightarrow \infty }T_{ye}(\{y\}_{1}^{N}, \tau
    ) = T_{x_{m}u_{m}}.
  \end{equation*}
\end{theorem}
\begin{proof}
  We have $\func{plim}_{N\rightarrow \infty }\hat{\mathbf{A}}_{\tau
  }(\{y\}_{1}^{N}) = \tilde{\mathbf{A}}_{\tau }^{\ast }$ from
  Proposition~\ref{prop:LSAR-mean-square}, which combined with~%
  \eqref{eq3:Pf_Theo3} implies
  \begin{equation*}
    \func{plim}_{N\rightarrow \infty }T_{ye}^{-1}(\tau )=T_{\tilde{x}%
      _{m}u_{m}}^{-1}(\tau ).
  \end{equation*}%
  Moreover, from Corollary~\ref{Theo2}, we have $ T_{\tilde{x}%
    _{m}u_{m}}(\tau )=T_{x_{m}u_{m}}$. The statement then follows
  from~\eqref{eq1:Pf_Theo3} and Lemma~\ref{Lemma3}.
\end{proof}

\section{Simulations}\label{sec:simulations}

In this section, we illustrate the performance of least-squares
auto-regressive estimation in identifying the manifest transfer
function in two examples, a deterministic directed ring network and a
group of Erd\H{o}s--R\'{e}nyi random networks. We pay particular
attention to the behavior displayed as the length of measured data and
the model order change.
In both examples, the input signal 
is a white Gaussian process with unit variance.

\begin{example}\longthmtitle{Directed ring network}\label{Example2}
  Consider a directed ring network of $40$ nodes with self-loops and
  all edge weights equal to $\alpha =0.25$. The nodes with indices
  $\{5, 23, 33, 34, 36\}$ are manifest and the remaining $35$ nodes
  are passive latent. Fig.~\ref{Fig_Ring_1e5}.(a) shows a 3D plot of the
  identification error $\lVert T_{ye}-T_{x_{m}u_{m}}\rVert _{\infty }$
  of the LSAR method, with axes corresponding to length of measured
  data and model order, respectively.
  We note that, when the measured data length $N$ is
  small, increasing the AR model order $\tau $ does not provide better
  estimation of the manifest transfer function. Similarly, when the
  model order $\tau $ is too low, increasing the data length $N$ is
  not helpful either.  Instead, when $N$\ and $\tau $ increase
  simultaneously, the LSAR method provides good estimation of the
  manifest transfer function without any knowledge of the latent
  nodes, as predicted by Theorem~\ref{Theo3}.  In
  Fig.~\ref{Fig_Ring_1e5}.(b), we fix 
  $N=10^{6}$ and show that the error of the model obtained by the LSAR
  method is quite similar to the error $\lVert T_{
    \tilde{x}_{m}u_{m}}-T_{x_{m}u_{m}}\rVert _{\infty }$ of the ideal
  AR model from Theorem~\ref{Theo1}.
  \begin{figure}[tbh]
    \centering
    \subfigure[]{
    \includegraphics[width=.47\linewidth]{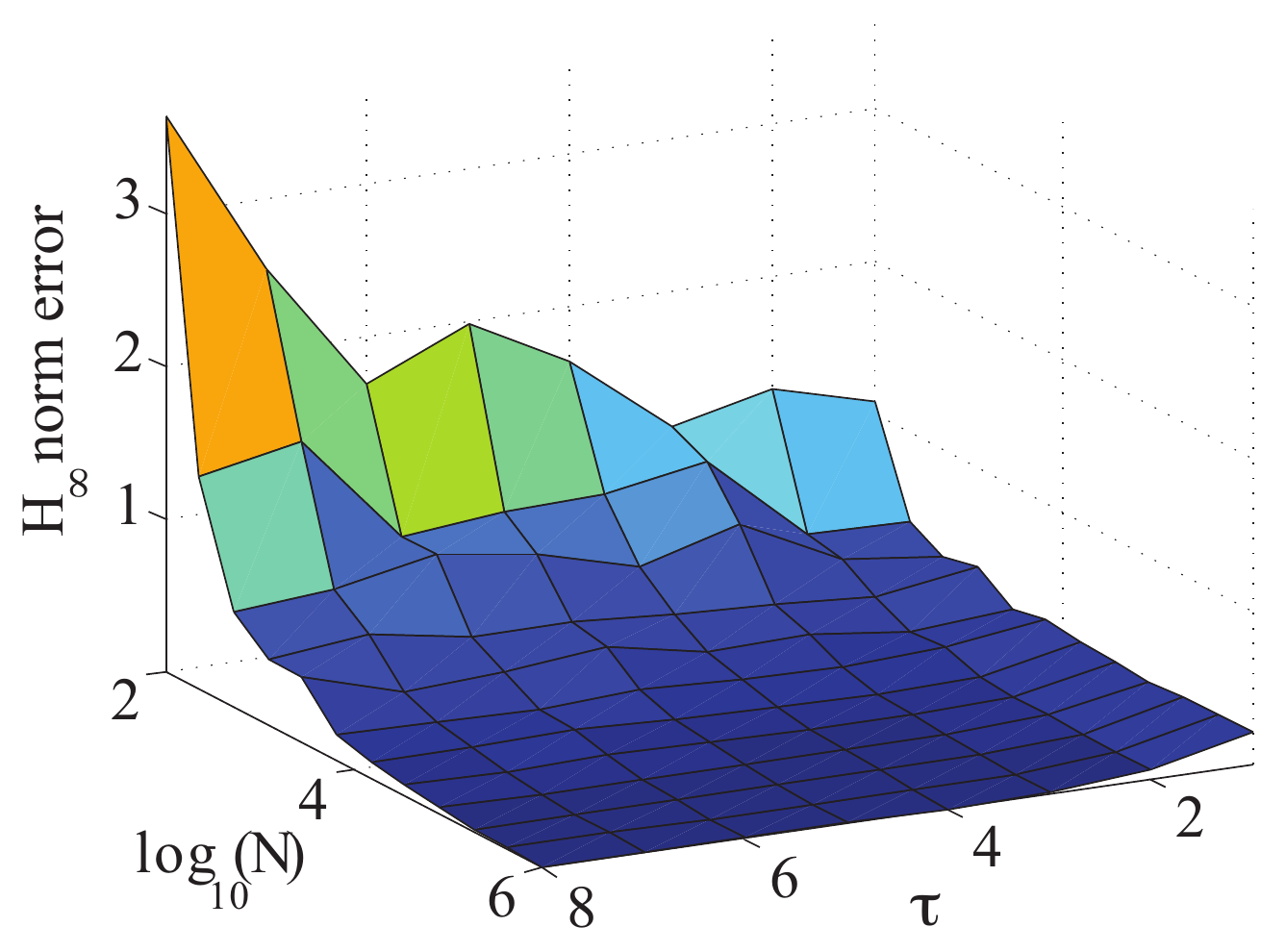}
    }
    \subfigure[]{
    \includegraphics[width=.46\linewidth]{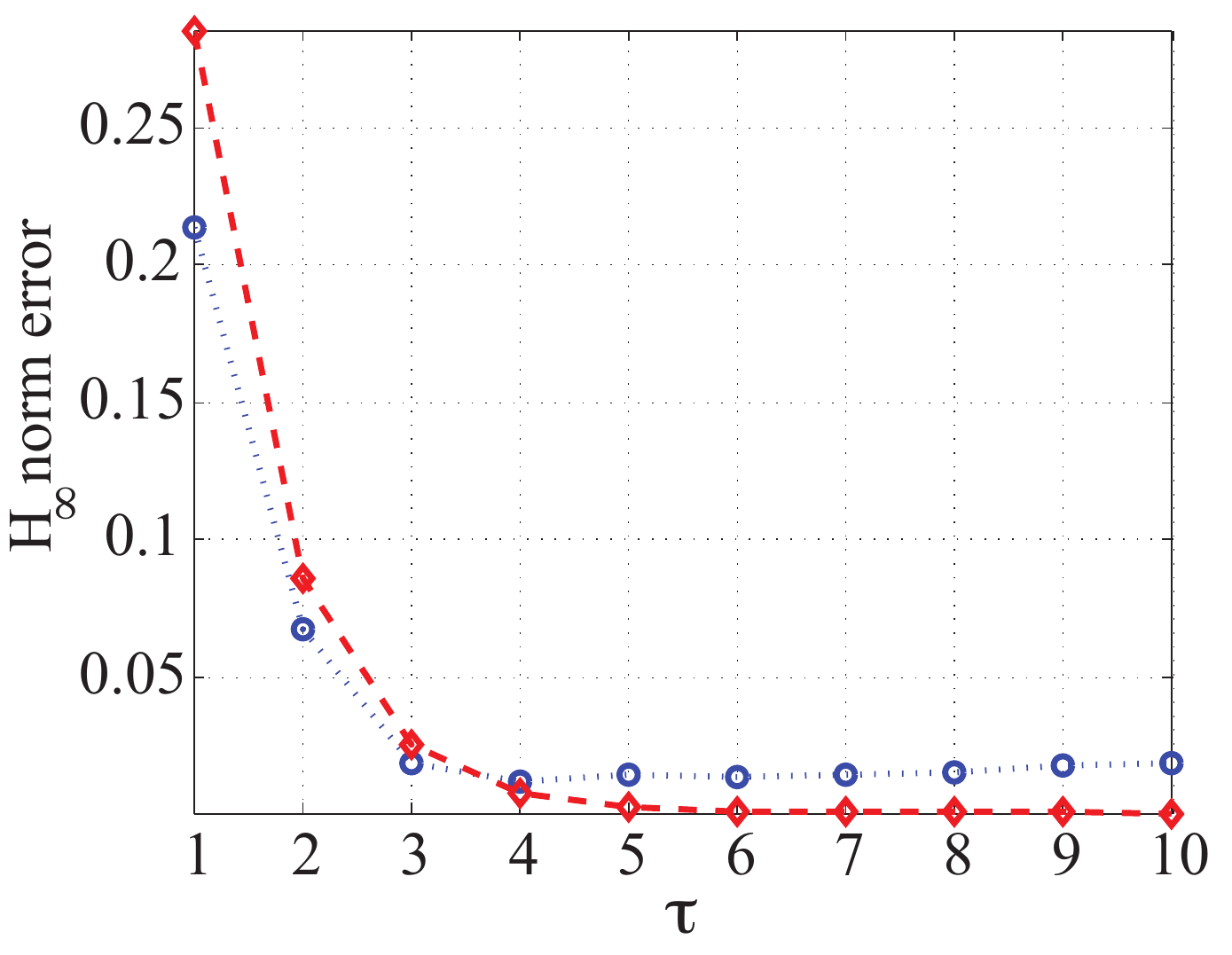}
    }
    \caption{\nrev{$H_{\infty }$-norm errors for the directed ring network of
      Example~\ref{Example2}. (a) The $H_{\infty }$-norm error of the LSAR
      method as a function of data length $N$ and model
      order $\tau$. Performance improves as $N$ and
      $\tau$ increase. (b) Comparison of the $H_{\infty }$-norm errors of the LSAR
      method (blue dotted lines) and the optimal AR model from
      Theorem~\ref{Theo1} (red dashed lines) for $N=10^{6}$.}}
      \label{Fig_Ring_1e5}
  \end{figure}
  Note that the latter requires knowledge of the true adjacency
  matrix~$A$, and we use it here merely for comparison purposes.
  \oprocend
\end{example}

\begin{example}\longthmtitle{Erd\H{o}s--R{\'{e}}nyi random
    network}\label{Example3}
  Here we consider a group of $10$ Erd\H{o}s--R{\'{e}}nyi random
  networks~\cite{BB:01}.  Each network in the group is of type
  $G(10,0.35)$, with $5$ manifest nodes chosen randomly and the
  remaining $5$ nodes are latent. Each pair of edges $(i,j),
  (j,i)$, $1\leq i <j\leq 10$ has nonzero weights with probability
  $0.35$ (we choose edges in pairs so that, when plotting the graph,
  the edge direction can be omitted). The weight of each edge has a
  uniform distribution in $\setdef{x\in \mathbb{R}}{0.1<x<0.35}$ (note
  that $(i,j)$ and $(j,i)$ can have different weights).  Because of
  rounding errors in the numerical computation, the estimated
  coefficient matrices~\eqref{eq:ls_AR} of the AR model are usually
  full matrices.  The lower bound on the edge weights allows us to
  discard entries in $\hat{\mathbf{A}}_{0}$ that are smaller than
  $0.1$.
  We consider a fixed length $N=10^{6}$ of measured data
  and analyze the effect of varying model
  order. Fig.~\ref{Fig:ER_3d} shows a 3D plot of the error in the
  identification of the manifest transfer function by the LSAR
  estimation, with axes corresponding to network index and model
  order, respectively.
  \begin{figure}[tbh]
    \centering
    \includegraphics[width=.65\linewidth]{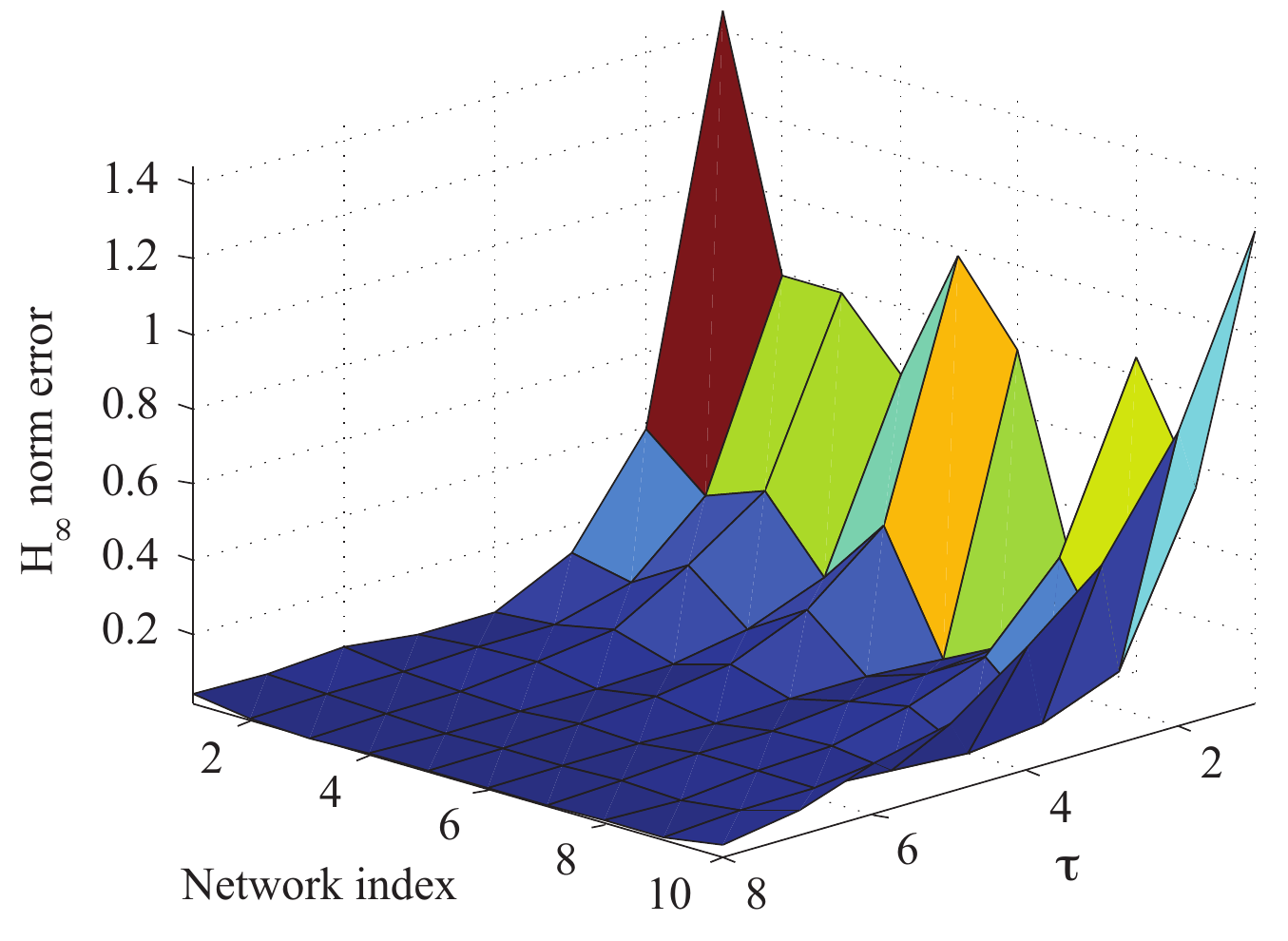}
    \caption{Illustration of the $H_{\infty }$-norm error of the LSAR
      with respect to the model order $\protect\tau $ for the group of
      $G(10,0.35)$ Erd\H{o}s--R{\'{e}}nyi random networks of Example
      \protect\ref{Example3}.  Performance improves as the model order
      $\protect\tau $ increases for all $10$ networks. The length of
      measured data is $N=10^{6}$.}\label{Fig:ER_3d}
  \end{figure}
  One can see the improvement in performance as the model order
  increases for all~$10$ networks.
  Fig.~\ref{Fig:ER_HinfErr_4models} compares the identification
  error of the LSAR method for the networks with indices $1, 6, 8, 10$
  in Fig.~\ref{Fig:ER_3d} against the error of the optimal AR model
  from Theorem~\ref{Theo1}.
  \begin{figure}[tbh]
    \centering
    \includegraphics[width=.75\linewidth]{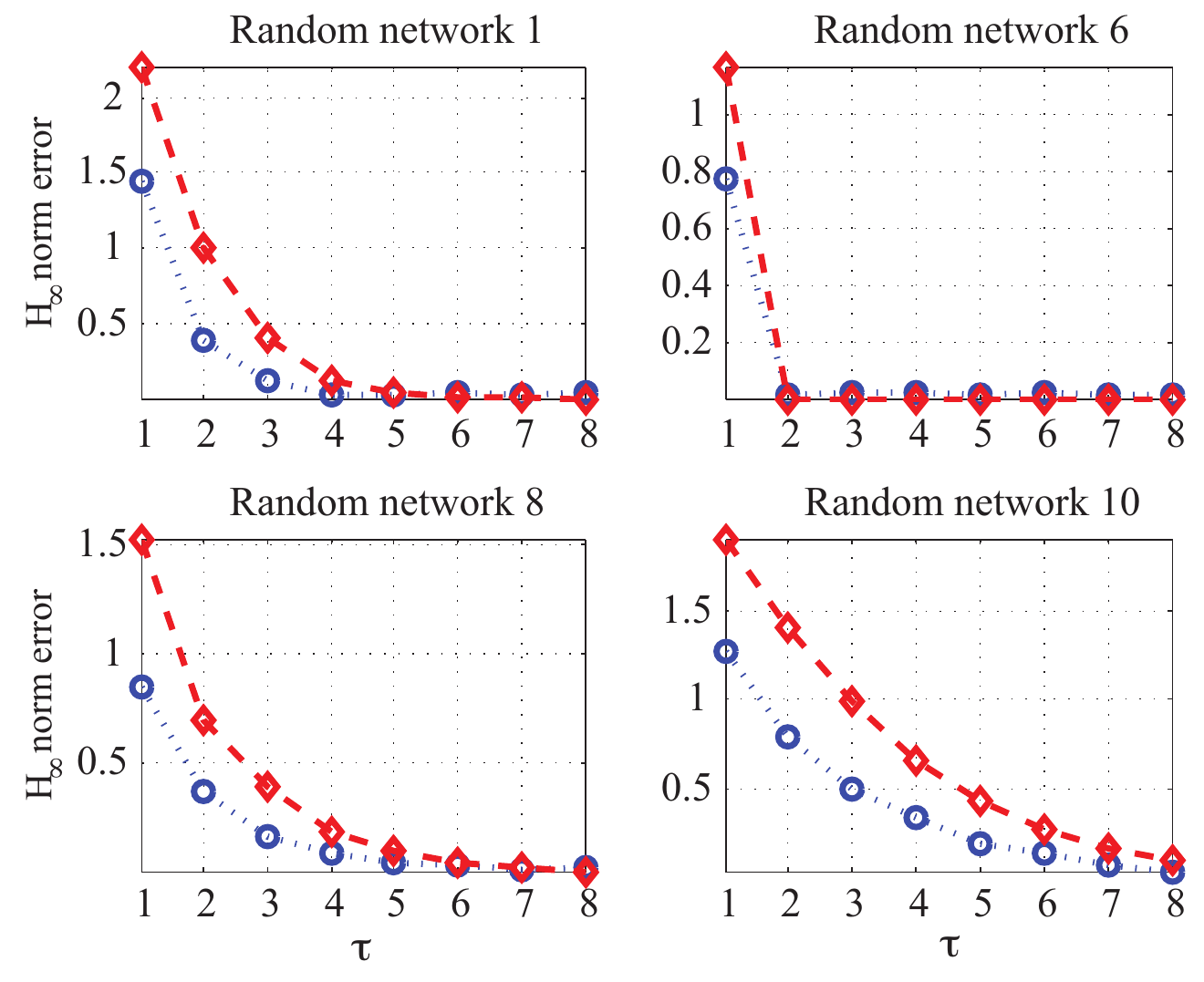}
    \caption{Comparison of the $H_{\infty }$-norm errors of the LSAR
      method (red dashed lines) and the optimal AR model from
      Theorem~\ref{Theo1} (blue dotted lines) for the
      Erd\H{o}s--R{\'{e}}nyi random networks with indices $1, 6, 8,
      10$ in Fig.~\ref{Fig:ER_3d}. The estimation error for network
      $6$ becomes $0$ when the AR model has order higher than $1$
      because the latent subnetwork is acyclic with $\tau_{22}=1$. The
      length of measured data is
      $N=10^{6}$. }\label{Fig:ER_HinfErr_4models}
  \end{figure}
  The latent subnetwork of network $6$ is acyclic (with
  $A_{22}=\mathbf{0}_{5 \times 5}$), and the estimation error goes to
  $0$ when the AR model has order higher than $\tau_{22}=1$, as
  predicted by Theorem~\ref{Theo3_nil}.  To illustrate our
  observations in Remark~\ref{re:LSAR-direct-latent} regarding the
  identification of manifest and latent interactions,
  Fig.~\ref{Fig:ER_graph} shows on the left the networks with
  indices $1, 6, 8, 10$ of Fig.~\ref{Fig:ER_3d} and on the right the
  corresponding reconstructions obtained with the LSAR method.  The
  indirect interactions represented by dashed edges in these plots
  imply the presence of latent nodes.  For comparison, we have also
  used the brain connectivity estimator technique called direct
  directed transfer function (dDTF)
  measure~\cite{MK-MD-WAT-SLB:01,AK-MM-MK-KJB-SK:03} from neuroscience
  to identify direct connections between nodes.  This technique is a
  refinement of the directed transfer function (DTF) approach, which
  instead cannot distinguish between direct and indirect connections.
  We have employed the dynamical modeling method within the Source
  Information Flow Toolbox
  (SIFT)~\cite{TM-AD-CK-SM:10,AD-TM-CK-ZAA-NBS-AV-SM:11} in
  EEGLAB~\cite{AD-SM:04}, which is a widely used open-source toolbox
  for EEG analysis.  Fig.~\ref{Fig:dDTF_ER_Network10} shows the
  interaction topology among the $5$ manifest nodes in network $10$
  identified by SIFT using the dDTF measure.  The dDTF measure is in
  the frequency domain and can also be a function of time (e.g., for
  time-varying networks). Since our networks are time-invariant, the
  time axis can be ignored. The plot shows that the dDTF identifies
  roughly equally strong connections for $(2,4)$ (which is in reality
  mediated by latent nodes) and $(4,5)$ (which is a true direct
  connection). This is in contrast with the identification made with
  the LSAR method presented in Fig.~\ref{Fig:ER_graph}(d).\oprocend
\end{example}

\begin{figure}[tbh!]
  \centering
  \subfigure[]{\includegraphics[width=.47\linewidth]{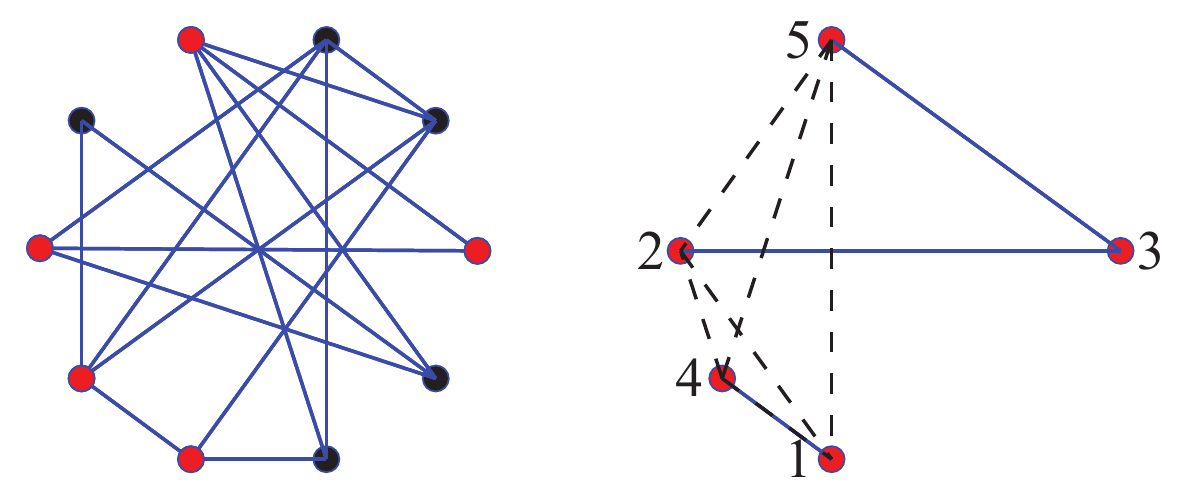}}
\hspace{10pt}%
  \subfigure[]{\includegraphics[width=.47\linewidth]{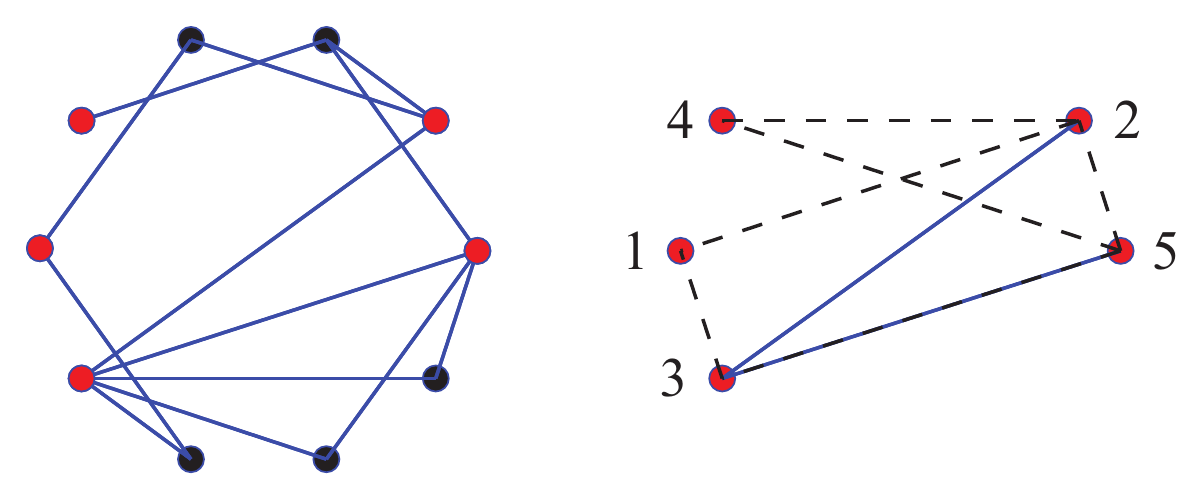}}
  \\[0.5ex]
  \subfigure[]{\includegraphics[width=.47\linewidth]{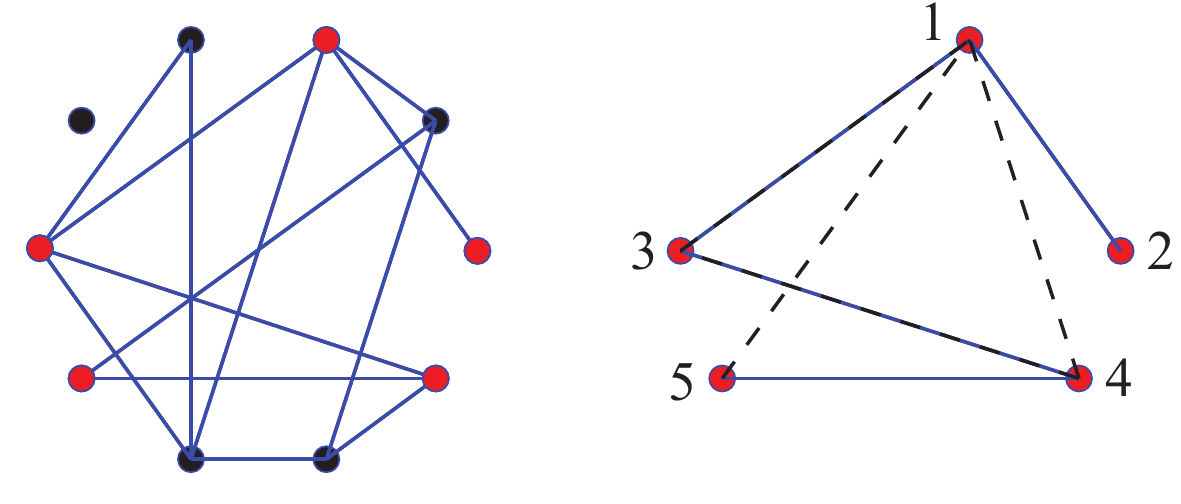}}
\hspace{10pt}%
  \subfigure[]{\includegraphics[width=.47\linewidth]{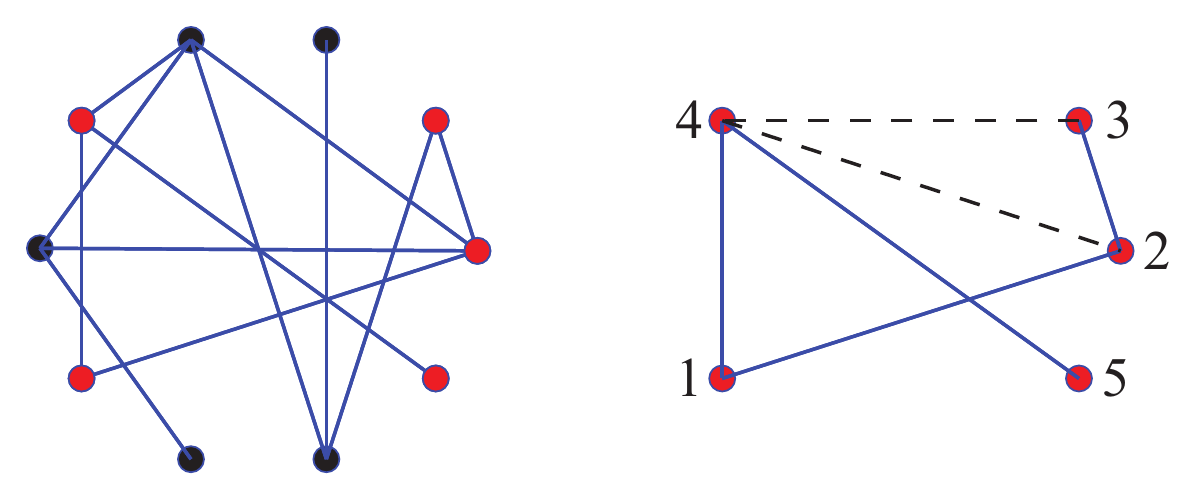}}
  \caption{Left: Erd\H{o}s--R{\'{e}}nyi random networks corresponding
    to the networks with indices $1$ (a), $6$ (b), $8$ (c), $10$ (d)
    in Fig.~\ref{Fig:ER_3d}, where red circles represent manifest
    nodes and black circles represent latent nodes. Right:
    reconstructed interaction graphs of the manifest subnetworks using
    the LSAR method. The numbers next to these nodes indicate their
    indices. A blue solid edge represents direct interaction and a
    black dashed edge represents indirect interaction through latent
    nodes. Note that the latent subnetwork of network $6$ is
    acyclic.}\label{Fig:ER_graph}
\end{figure}

\begin{figure}[tbh!]
  \centering
  \subfigure[]{\includegraphics[width=.8\linewidth]{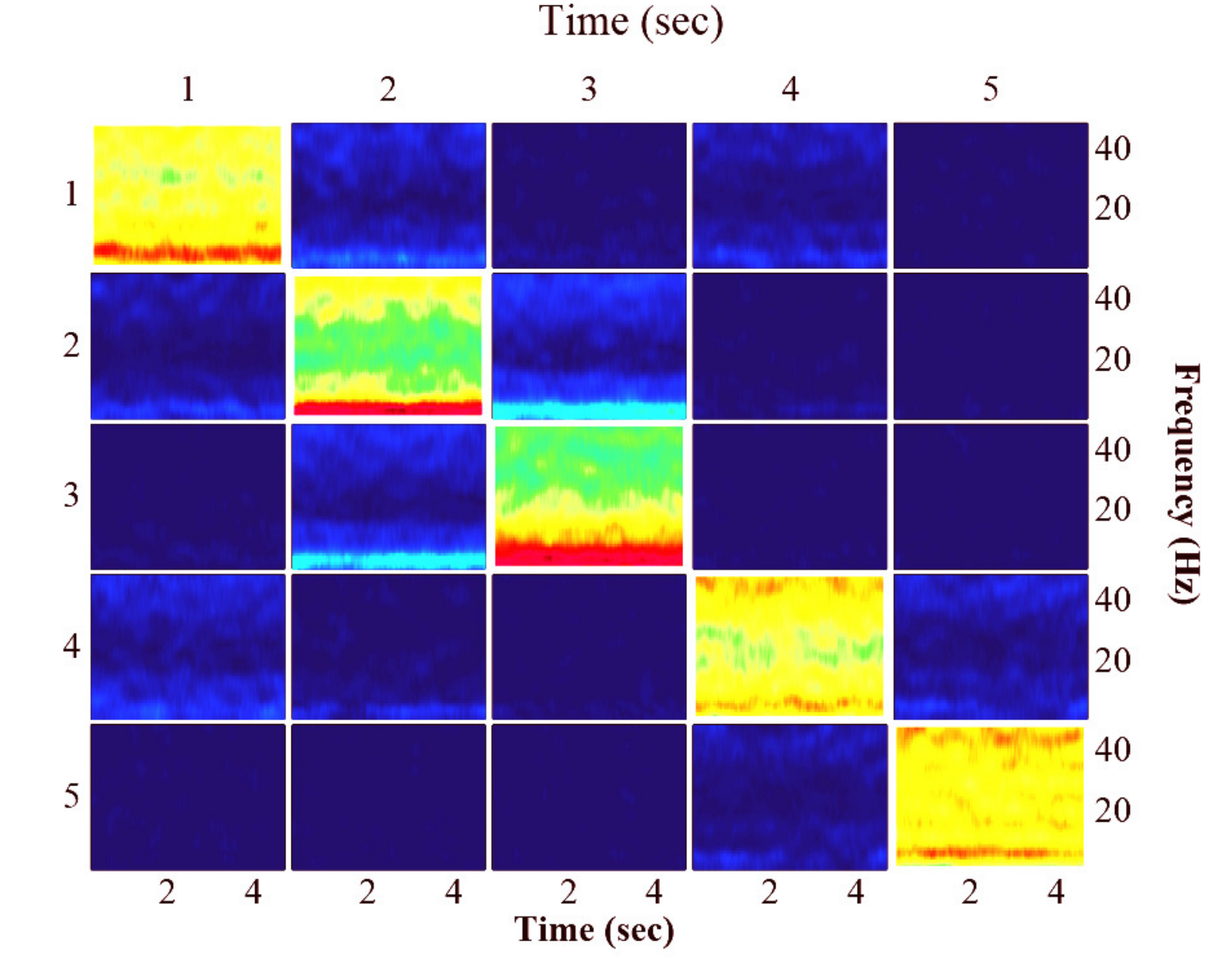}}
  \\
  \subfigure[]{\includegraphics[width=.8\linewidth]{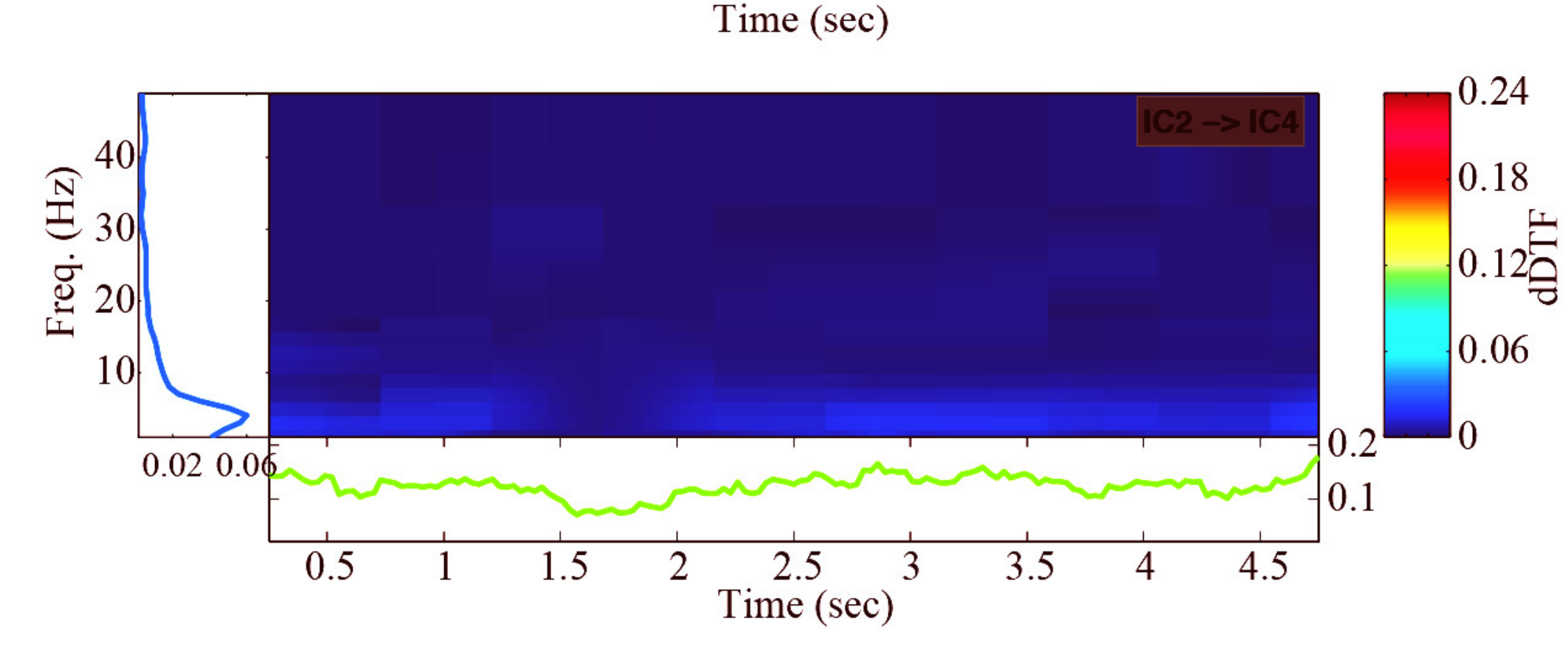}}
  \\
  \subfigure[]{
  \includegraphics[width=.8\linewidth]{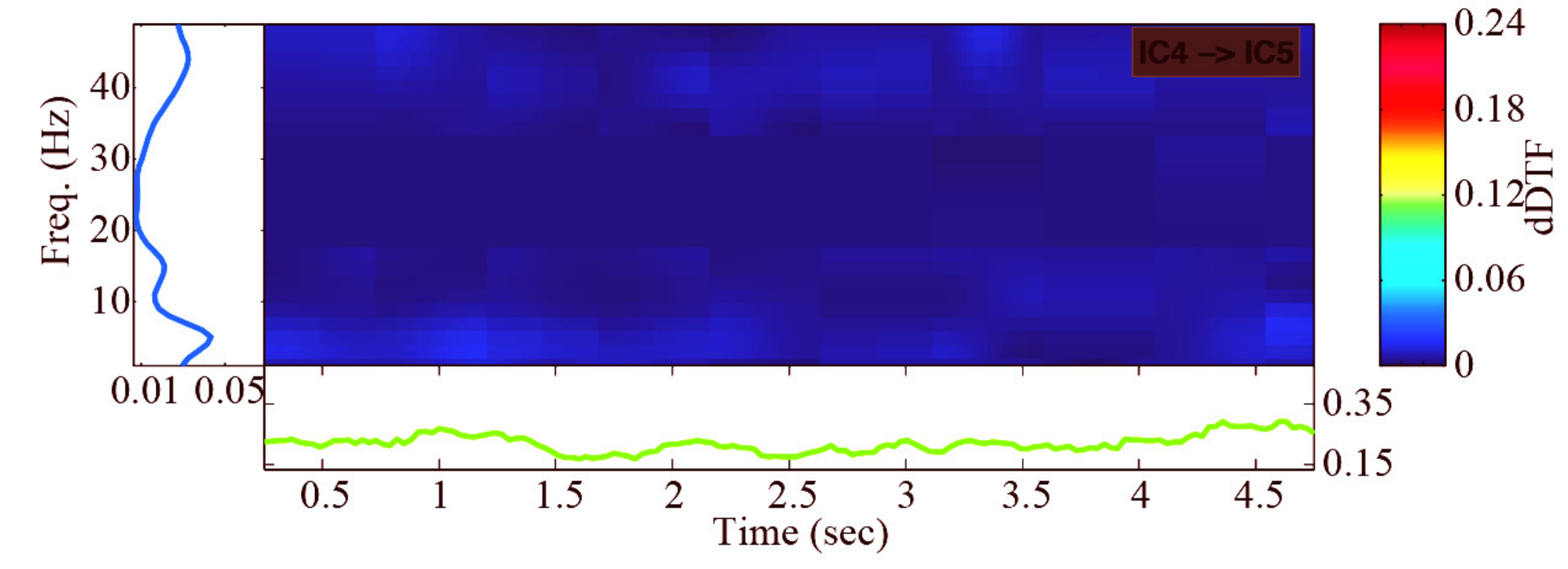}}
  \caption{(a) The interaction topology identified by the dDTF
    method for the Erd\H{o}s--R{\'{e}}nyi network with index $10$. (b and c) A zoom-in of the (indirect) connection
    $(2,4)$ and the (direct) connection
    $(4,5)$, resp.}\label{Fig:dDTF_ER_Network10}
  \vspace*{-1.5ex}
\end{figure}

\rev{
  \begin{example}\longthmtitle{Cortical brain network identification
      from EEG data}\label{ex:eeg}
    In this example, we apply our method to a multi-channel
    electroencephalogram (EEG) time-series recorded from a human scalp
    during a selective visual attention experiment in order to
    identify the manifest and latent-mediated connections among the
    channels. The EEG data is taken from the sample dataset available
    in the EEGLAB MATLAB toolbox~\cite{AD-SM:04}. This dataset
    contains recordings from 32 channels for more than $3$ seconds
    with \nrev{$T_s = 7.8$\! ms} sampling time ($128$\! Hz sampling
    frequency).  \nrev{Channel locations are shown in
      Fig.~\ref{fig:eeg-loc-R}(a) on a top (axial) view of the
      skull.  During the experiment, the subject is asked to perform
      specific motor actions in response to certain visual stimuli,
      requiring coordination among several cortices. We take the first
      13 EEG channels corresponding to the fronto-temporal cortical
      areas (shown as blue squares in Fig.~\ref{fig:eeg-loc-R}(a))}
    as the manifest nodes and the remaining channels as well as the
    truly hidden brain regions (the ones not probed in the test) as
    the latent nodes. In the following, we present the results of
    identifying the direct and indirect connections among the manifest
    nodes using the LSAR method as well as the dDTF
    algorithm~\cite{MK-MD-WAT-SLB:01,AK-MM-MK-KJB-SK:03} and the S+L
    algorithm of~\cite{MZ-RS:16}. For each method, we only keep the
    edges whose identified weights are above a certain threshold
    $\theta$ (which we choose as a proportional constant $\alpha \in
    (0, 1)$ times the largest edge weight in the network).

    In neuroscience, the brain dynamics generating the EEG data are
    usually approximated by a high-order AR model of the
    form~\eqref{eq:LTI_Network-HO} ($\nu \gtrsim 10$). As mentioned in
    Remark~\ref{rem:higher-order2}, larger $\tau$ and thus larger
    number of parameters are then required, which may lead to
    over-parametrization. To prevent this, we use an
    exponentially-regularized version of~\eqref{eq:ls_AR} by
    minimizing
    \begin{align}\label{eq:reg}
      \func{tr}(\vec e_N \vec e _N^T + \gamma \hat{\mathbf A}_\tau P
      P^T \hat{\mathbf A}_\tau^T),
    \end{align}
    where $P = \func{diag}(1, \rho_0^{-1}, \dots, \rho_0^{-(\tau -
      1)}) \otimes I_{n_m}$ and, ideally, $\rho_0 = \rho(A_{22})$ (in
    practice, it is found by trial and error). The role of the
    exponential regularizer is to encourage the higher-order AR terms to
    decay exponentially, as $\tilde A_i^*$ do. In the simulations that
    follow, we have used $\gamma = 10$ and $\rho_0 = 0.9$.
    
    Fig.~\ref{fig:eeg-graph-ours} shows the reconstructed manifest
    subnetwork with direct and indirect connections using the LSAR
    method for $\tau = 15$ and different values of $\alpha$. One can
    observe that the sensitivity of the network structure to the
    threshold ratio $\alpha$ is significant, showing that the majority
    of network links are relatively weak with respect to the largest
    link (which is usually a self-loop). This sensitivity, however, is
    smaller for the indirect connections. Note that increasing
    $\alpha$ is a way of enforcing sparsity among the manifest nodes
    similar (but not equivalent) to~\cite{MZ-RS:16}. Also, note that
    unlike~\cite{MZ-RS:16}, the manifest subnetwork estimated by our
    method is directed (though directions are not shown in
    Fig.~\ref{fig:eeg-graph-ours} for simplicity).
    \begin{figure}[htb!]
      \subfigure[$\alpha =
      0.15$]{\includegraphics[width=0.47\linewidth]{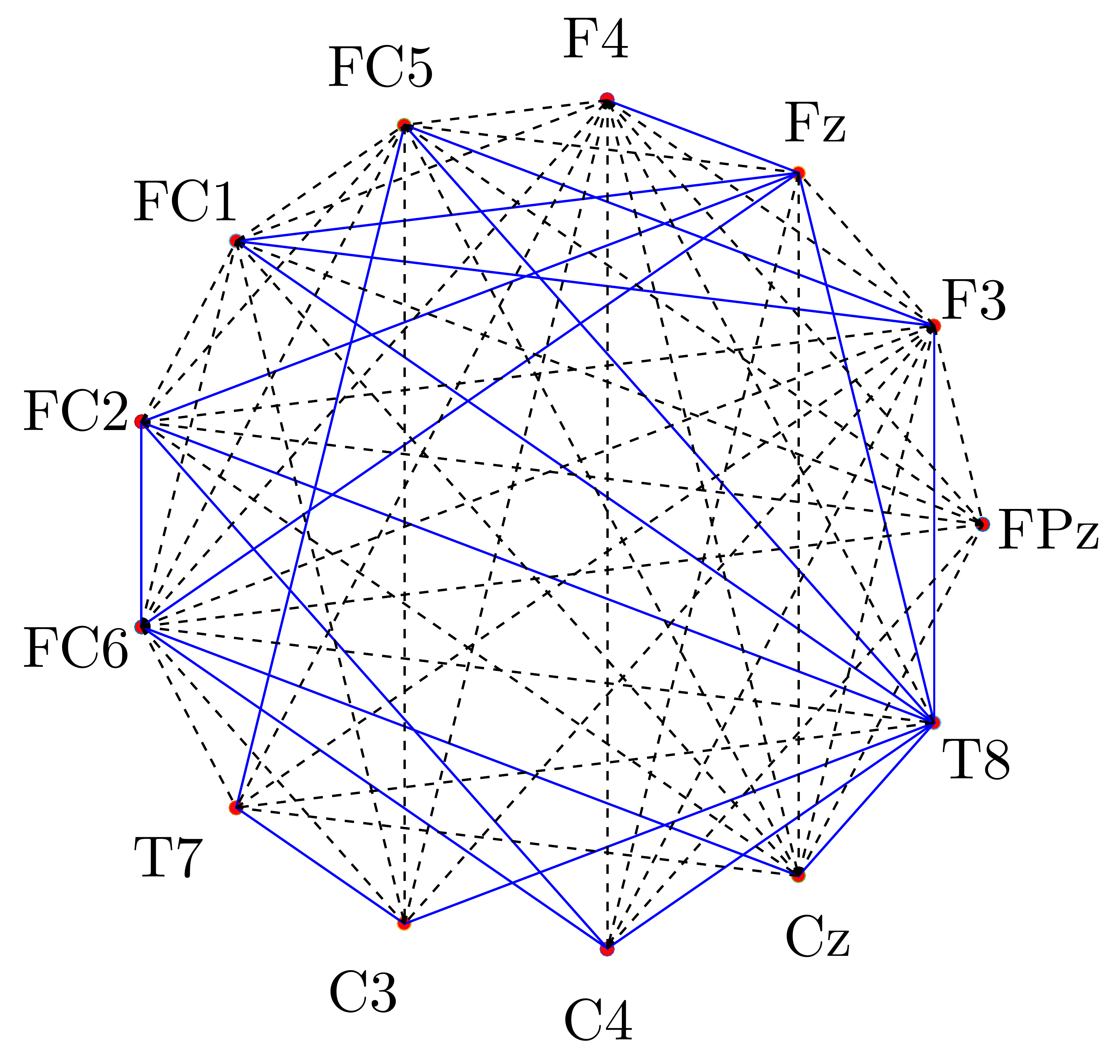}}
      \hspace{5pt}
      \subfigure[$\alpha =
      0.12$]{\includegraphics[width=0.47\linewidth]{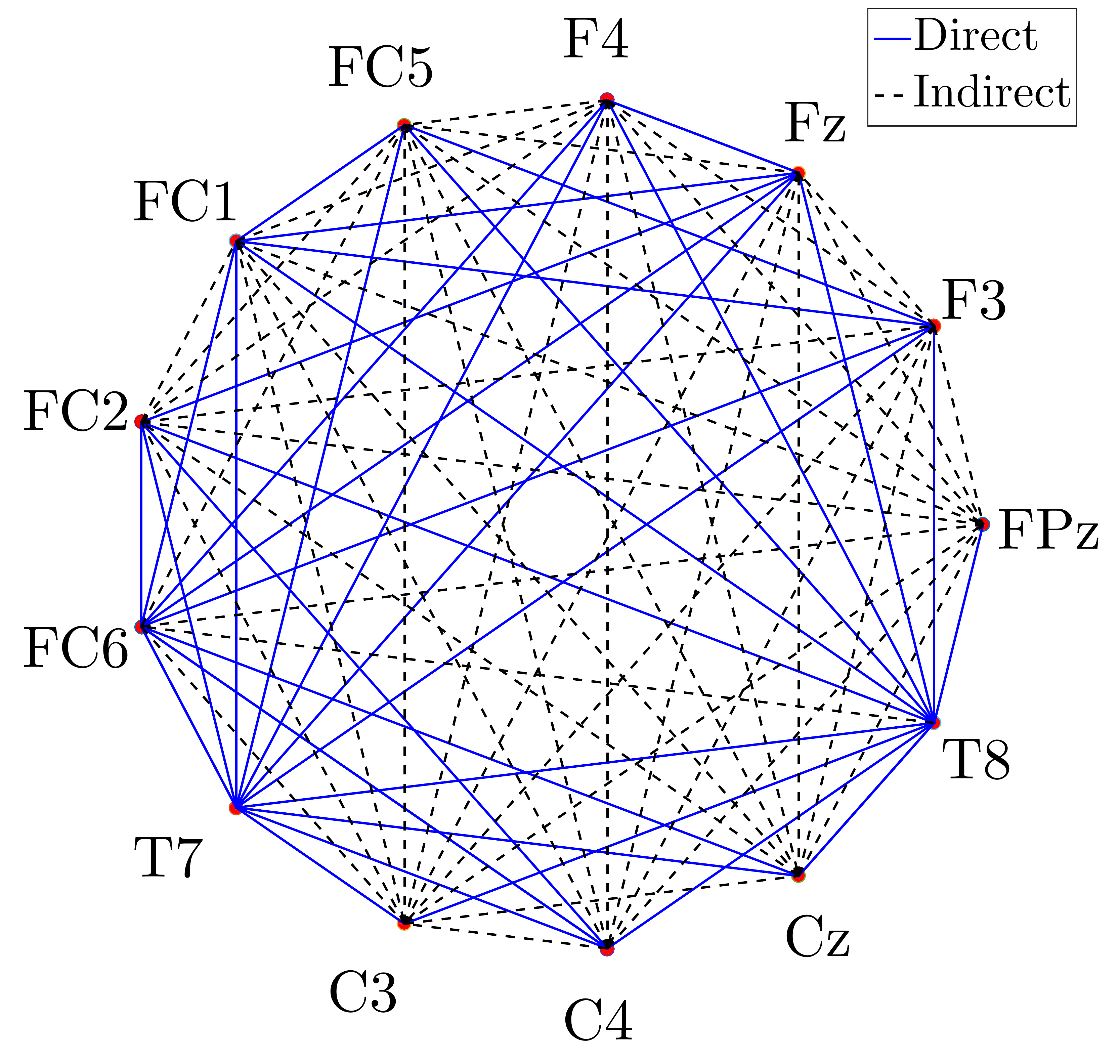}}
      \\
      \subfigure[$\alpha =
      0.09$]{\includegraphics[width=0.47\linewidth]{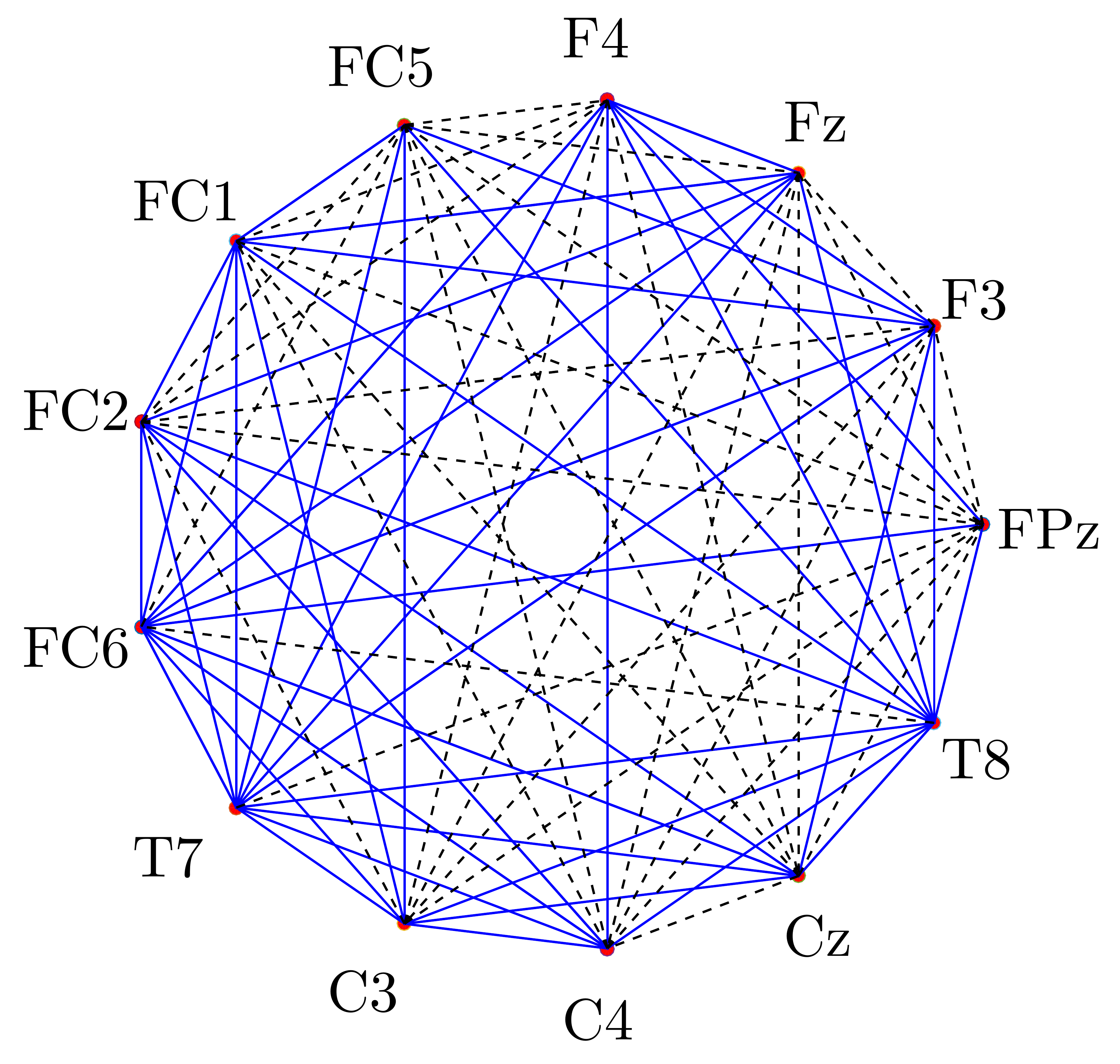}}
      \hspace{5pt}
      \subfigure[$\alpha =
      0.06$]{\includegraphics[width=0.47\linewidth]{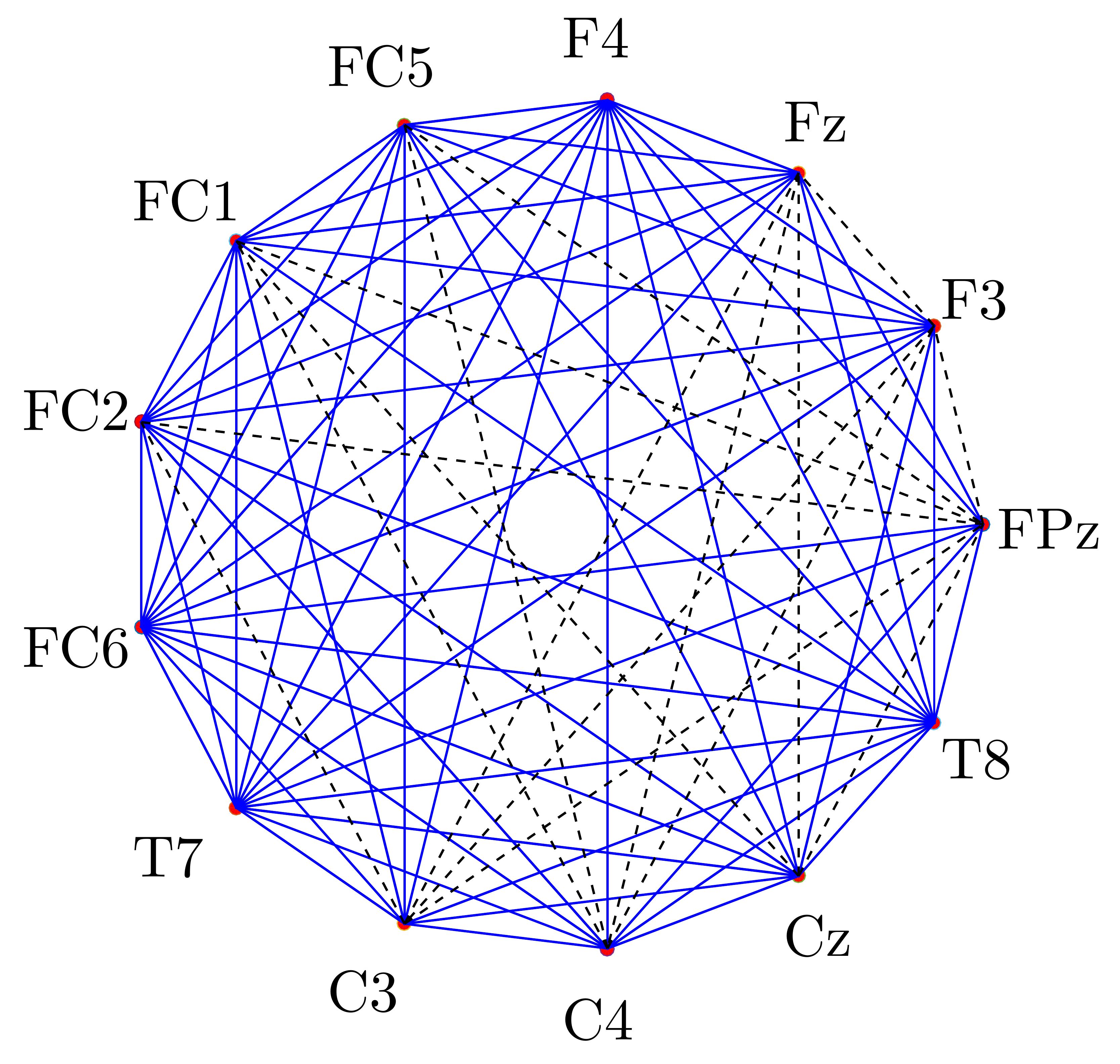}}
      \caption{Reconstructed manifest subnetwork for the EEG data
        in Example~\ref{ex:eeg} using our proposed method \nrev{with the exponentially-regularized objective function~\eqref{eq:reg} and $\gamma = 10$, $\rho_0 = 0.9$,} and $\tau =
        15$. The direct (solid blue) and indirect (dashed black)
        connections are depicted for different values of threshold
        ratio $\alpha$. For each value of $\alpha$, the connections
        whose weights are smaller than $\alpha$ times the largest
        network weight are removed.}
      \label{fig:eeg-graph-ours}
    \end{figure}
    
    For comparison, Fig.~\ref{fig:eeg-graph-ZS} shows the
    reconstructed manifest subnetwork with direct and indirect
    connections using the S+L method of~\cite{MZ-RS:16} for $n =
    5$\footnote{$n$ represents the model order
      in~\cite{MZ-RS:16}. While the role of the model order is not
      discussed in the reference, the use of higher-order models
      significantly increases the computational cost of the
      algorithm. Also, note that there is no one-to-one correspondence
      between the subfigures of
      Figs.~\ref{fig:eeg-graph-ours}-\ref{fig:eeg-graph-EEGLAB}.}. Although
    the use of a threshold value is not prescribed in~\cite{MZ-RS:16},
    we have used a fixed value of $\alpha = 0.01$ for all values of
    $(\lambda, \gamma)$, since the absence of a threshold ($\alpha =
    0$) results
    in all nodes being estimated to be (both directly and indirectly)
    connected. This lack of sparsity occurs for all values of
    $(\lambda, \gamma)$ (no matter how large they are chosen), unless
    extremely large values are employed, which results in a fully
    disconnected network. 
    From various plots, we see that even
    with the use of a threshold value all the nodes are
    estimated to be indirectly connected, with the sparsity of direct
    connections and the estimated number of latent nodes being determined by $(\lambda, \gamma)$.
    \nrev{This abundance of indirect connections and parameter-based tuning of direct connectivity is similar to our results in Fig.~\ref{fig:eeg-graph-ours}, even though the details of the reconstructed networks do not exactly match.}
    \begin{figure}[htb!]
      \subfigure[$\lambda = 1, \gamma\lambda = 0.1, l =
      1$]{\includegraphics[width=0.47\linewidth]{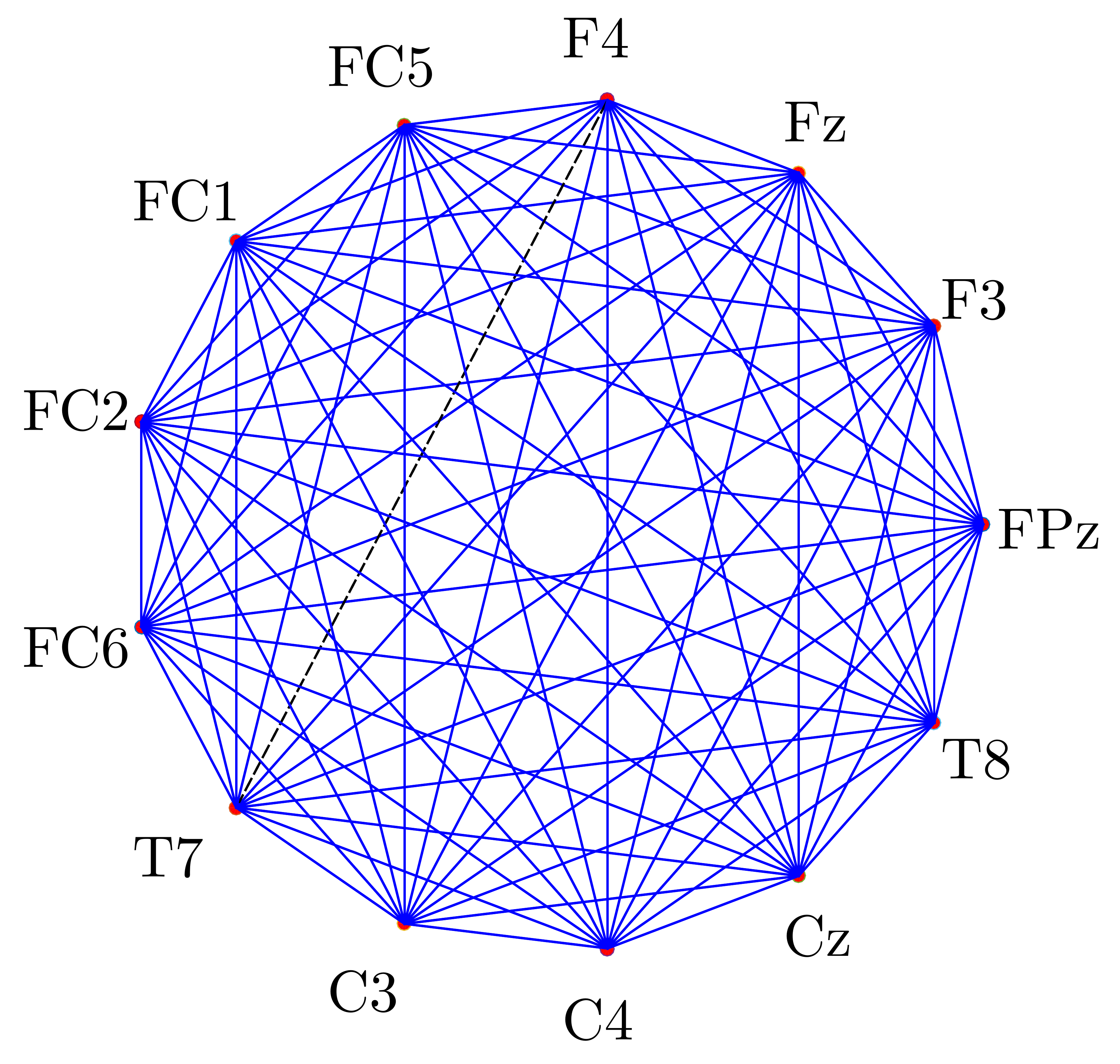}}
      \hspace{5pt}
      \subfigure[$\lambda = 1, \gamma\lambda = 0.5, l =
      13$]{\includegraphics[width=0.47\linewidth]{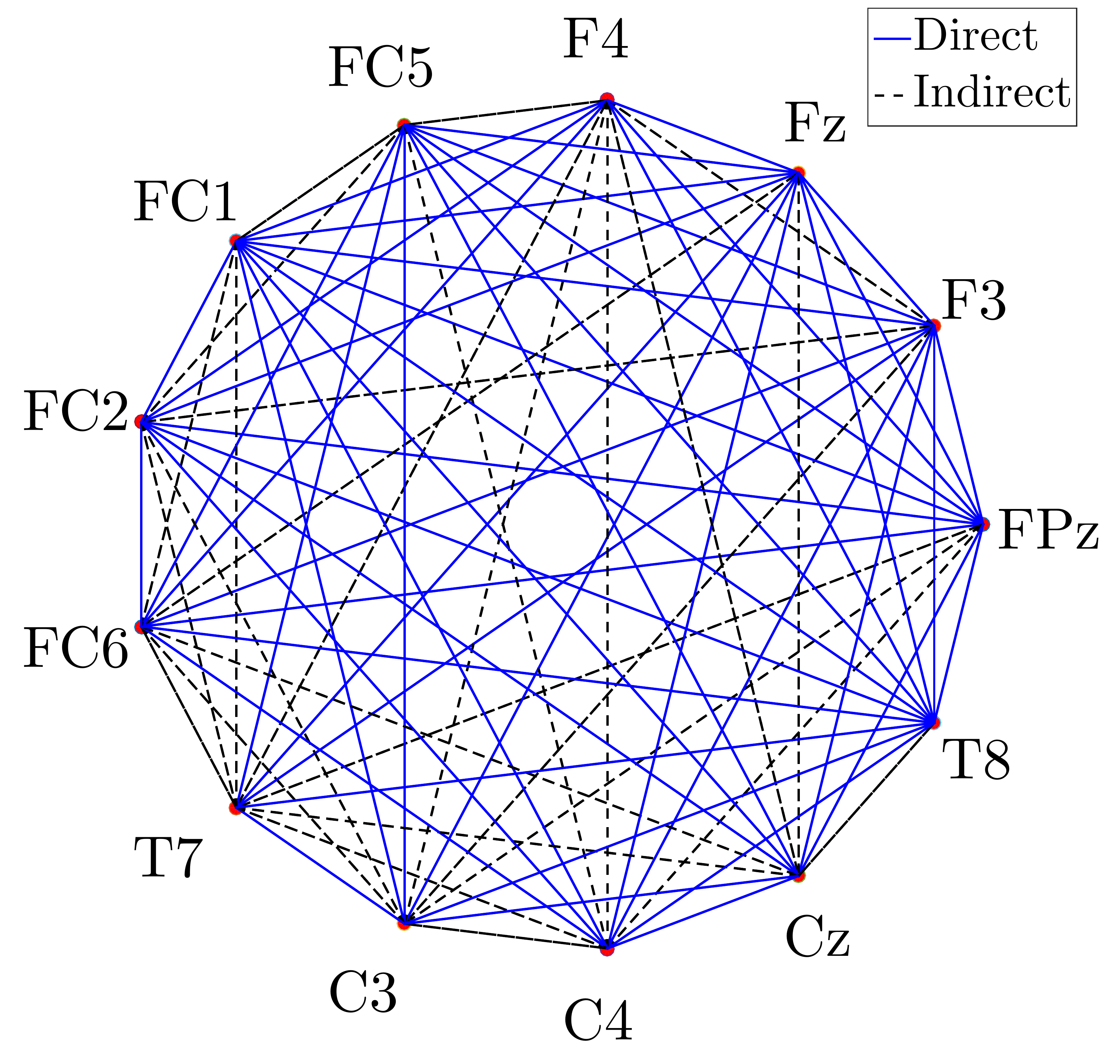}}
      \\
      \subfigure[$\lambda = 1, \gamma\lambda = 1, l =
      18$]{\includegraphics[width=0.47\linewidth]{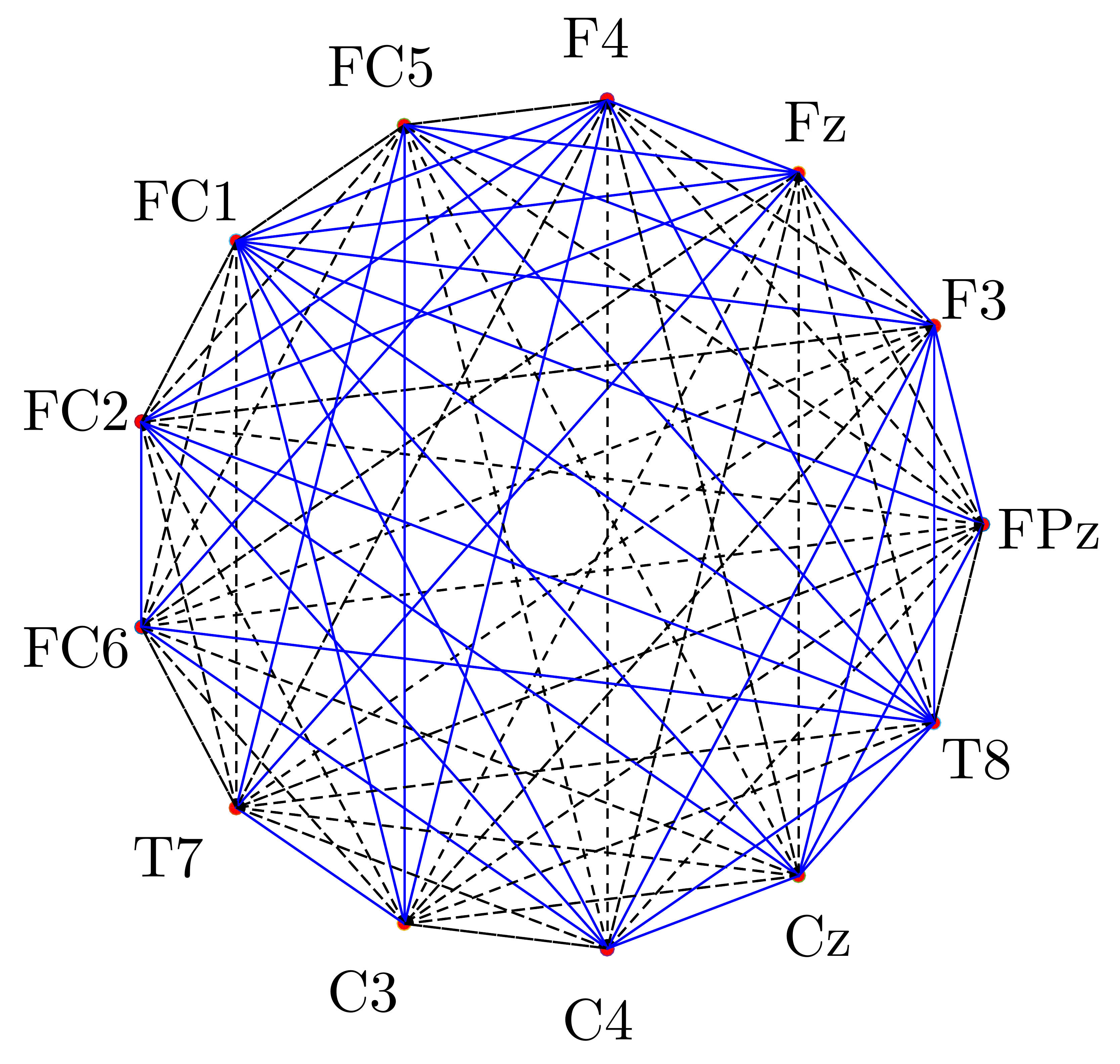}}
      \hspace{5pt}
      \subfigure[$\lambda = 0.1, \gamma\lambda = 0.5, l =
      50$]{\includegraphics[width=0.47\linewidth]{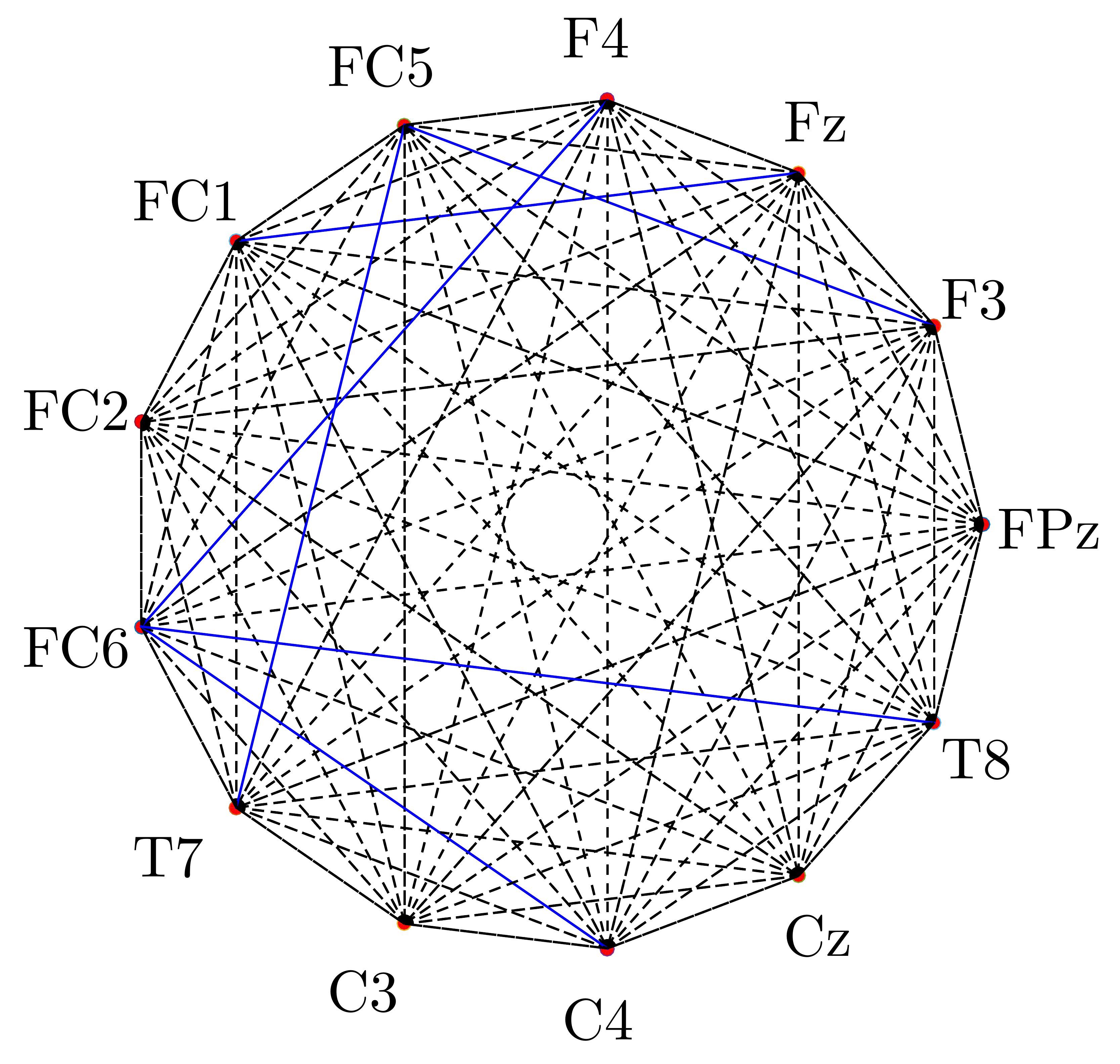}}
      \caption{Reconstructed manifest subnetwork for
        Example~\ref{ex:eeg} using the \nrev{S+L} method in~\cite{MZ-RS:16}. The
        direct (solid blue) and indirect (dashed black) connections
        are depicted for different values of weight parameters
        $(\lambda, \gamma)$ and fixed threshold ratio $\alpha =
        0.01$. $l$ represents the estimated number of latent
        nodes.}\label{fig:eeg-graph-ZS}
      \vspace*{-15pt}
    \end{figure}

    Fig.~\ref{fig:eeg-graph-EEGLAB} shows the result of applying
    both the Directed Transfer Function (DTF)~\cite{MJK-KJB:91} and
    direct Directed Transfer Function (dDTF) methods to the EEG
    channel data to estimate the indirect and direct connections
    between the manifest nodes, respectively, for different frequency
    bands. Both methods are applied to the data using the EEGLAB SIFT
    plugin for $\tau = 15$ (selected based on SIFT Model Order
    Selection). In all cases, a constant threshold ratio $\alpha =
    0.1$ is used and the value of the threshold is computed with
    respect to the largest \emph{off-diagonal} link weight in the same
    frequency. \nrev{As can be seen, the connectivity pattern is
      considerably different between lower and higher frequencies,
      where several pairs are not even indirectly connected
      over the $\delta$-$\theta$ band. This is in contrast to the
      reconstructed networks of Fig.~\ref{fig:eeg-graph-ours} in
      which most pairs are at least indirectly connected, even for
      threshold values as large as $\alpha = 0.15$. Nevertheless, a
      common feature of all the reconstructed networks in
      Figs.~\ref{fig:eeg-graph-ours}-\ref{fig:eeg-graph-EEGLAB} is
      that the density of direct connections is higher in the
      fronto-central (FC) areas and lower in central (C) areas and
      midline frontal pole (FPz).
      The independence of this sparsity pattern from the employed
      reconstruction method and parameter value suggests that it is a
      robust feature of the actual brain connectivity among these
      areas.}
    
    \begin{figure}
      \subfigure[$f = 1^{Hz}$ ($\delta$
      band)]{\includegraphics[width=0.47\linewidth]{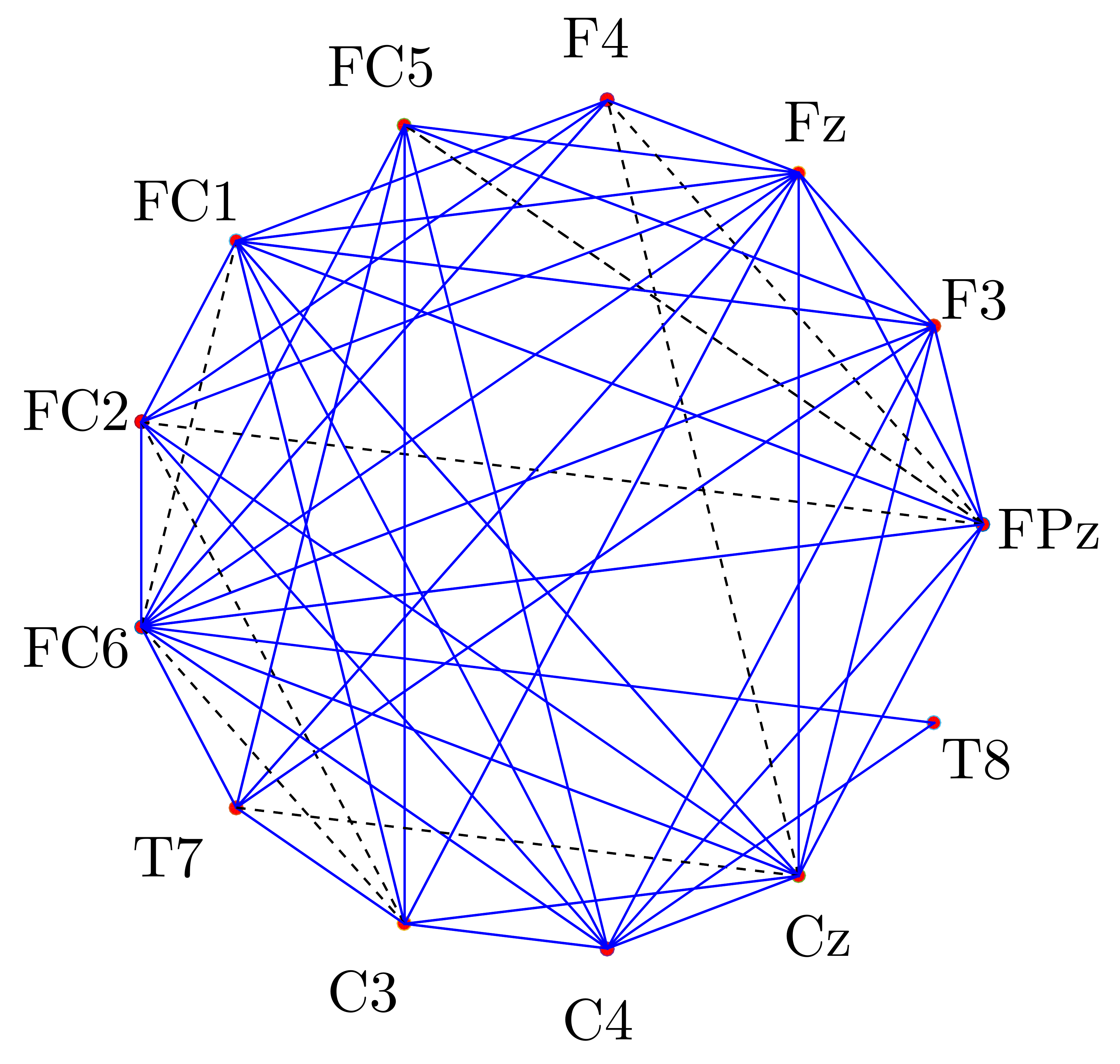}}
      \hspace{5pt}
      \subfigure[$f = 5^{Hz}$ ($\theta$
      band)]{\includegraphics[width=0.47\linewidth]{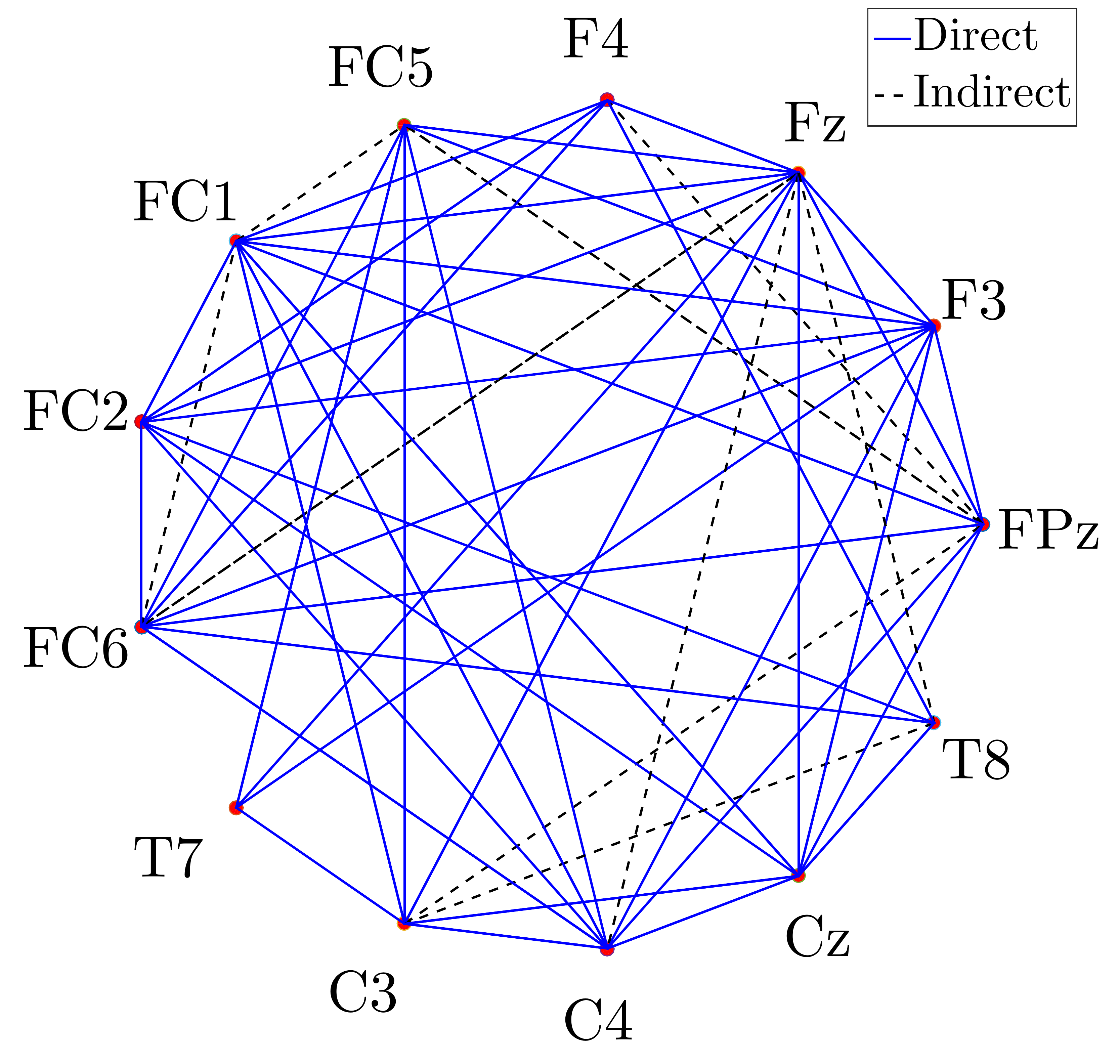}}
      \\
      \subfigure[$f = 20^{Hz}$ ($\beta$
      band)]{\includegraphics[width=0.47\linewidth]{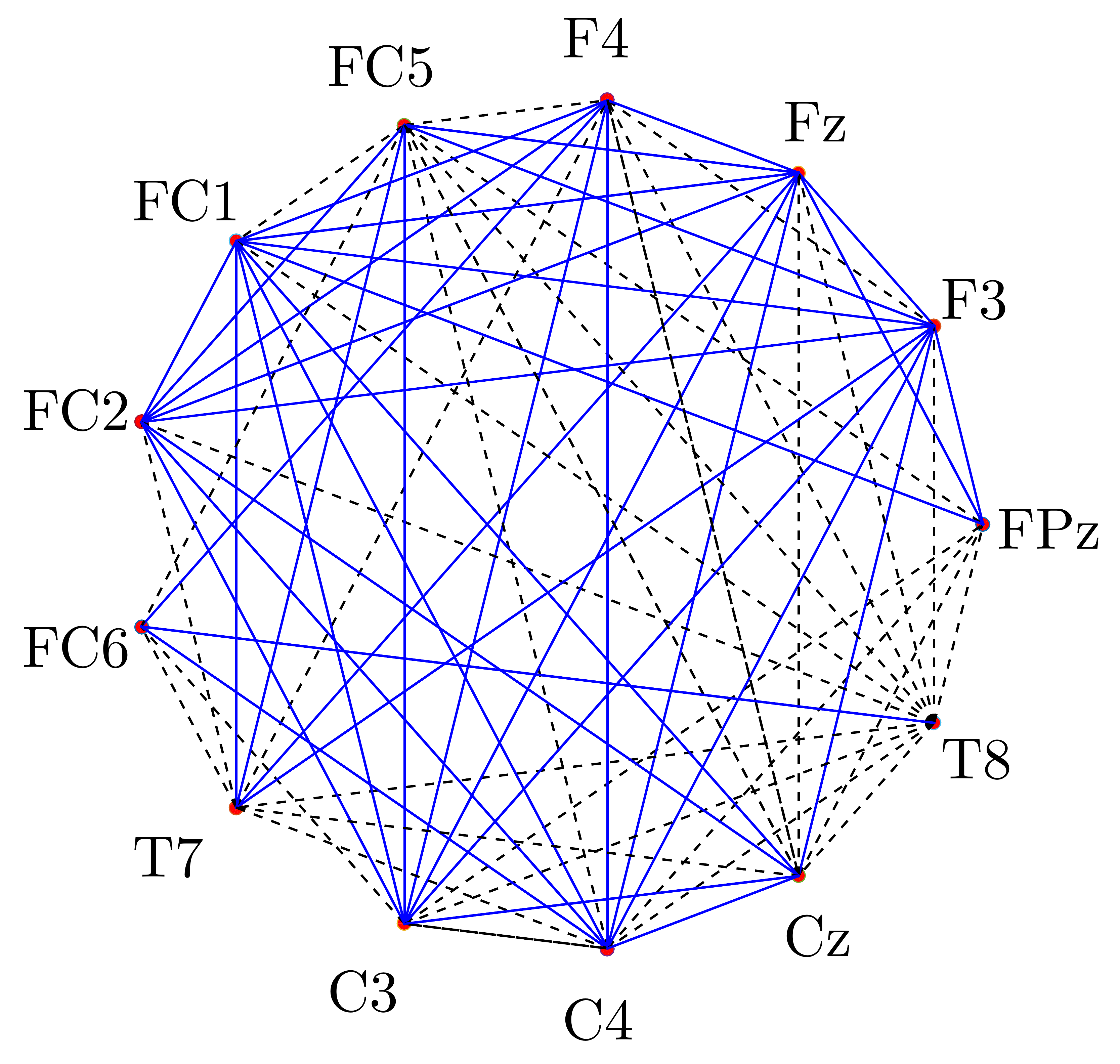}}
      \hspace{5pt}
      \subfigure[$f = 50^{Hz}$ ($\gamma$
      band)]{\includegraphics[width=0.47\linewidth]{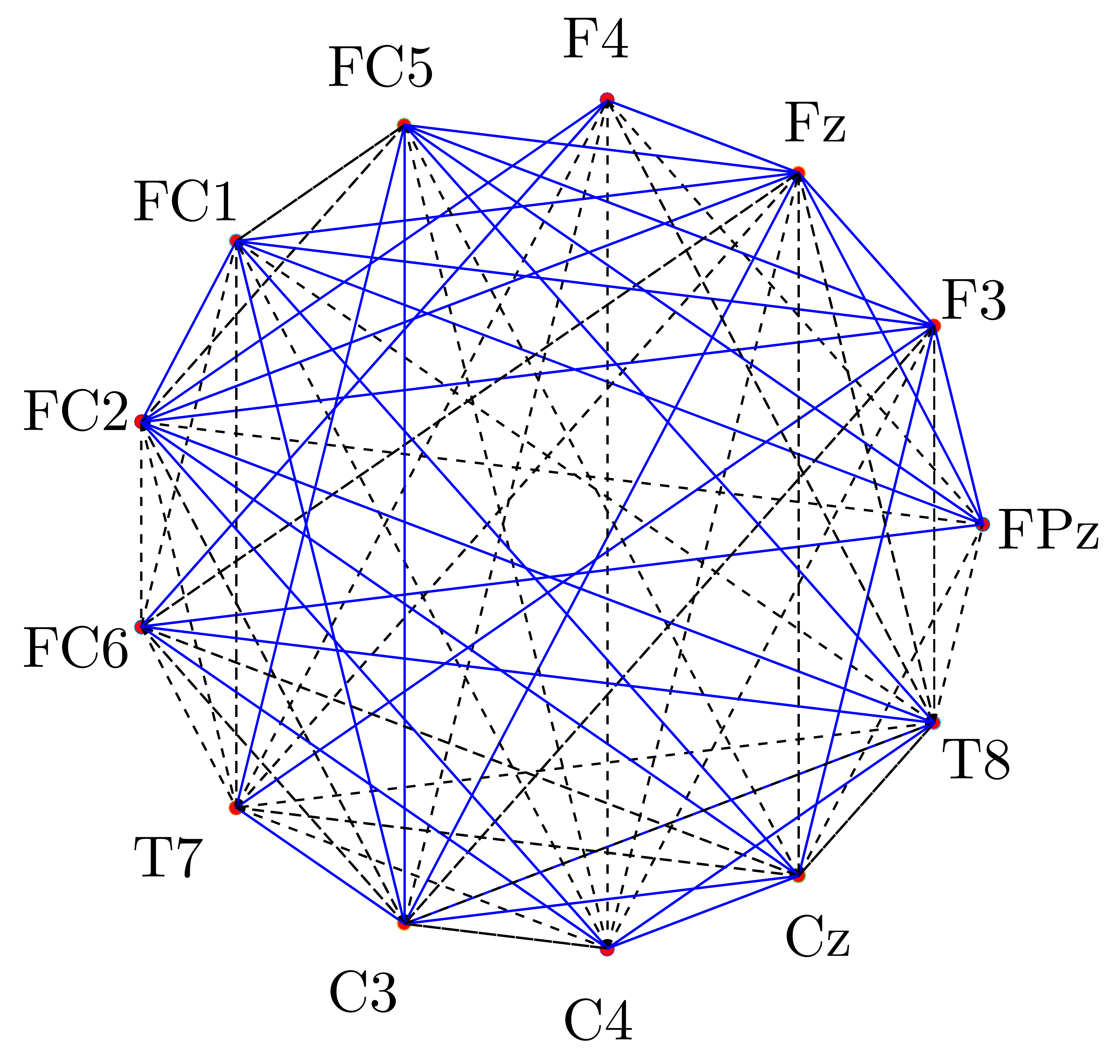}}
      \caption{Reconstructed manifest subnetwork for
        Example~\ref{ex:eeg} using the combination of DTF and dDTF
        estimation methods. The direct (solid blue) and indirect
        (dashed black) connections are illustrated for different
        frequency values and fixed threshold ratio $\alpha = 0.1$.}
      \label{fig:eeg-graph-EEGLAB}
      \vspace*{-10pt}
\end{figure}

Since the \emph{true} network structure is unknown for this example (and
hence the methods are not directly comparable), we validate our LSAR
estimated connectivity % 
based on its ability to predict
\emph{future} (i.e., \nrev{unseen}) channel activity.  Thus, we used
the first \nrev{$80\%$} of data for LSAR estimation
and the last \nrev{$20\%$} for evaluation, which is based on
\begin{align}\label{eq:R}
  R^2 = 1 - \frac{\sum_{k = N+1}^{N'} \|e(k)\|^2}{\sum_{k = N+1}^{N'}
    \|y(k)\|^2},
\end{align}
denoting the percentage of the future channel activity that is
correctly predicted by the model~\cite[\S 16.4]{LL:99}, where
$\{y(k)\}_{k = N+1}^{N'}$ is the latter data sequence not used for
estimation. \nrev{The blue curve in Fig.~\ref{fig:eeg-loc-R}(c)
  shows the value of \nrev{$R \times 100\%$} for the LSAR method as a
  function of model order for the same selection of nodes as above
  (i.e., anterior)%
  \footnote{Edge values are not thresholded ($\alpha = 0$) for
    computing $R$ values.}.  This shows that the method is capable of
  predicting more than \nrev{$96.5\%$} of \nrev{unseen} data with
  model orders $\tau = 15 \sim 20$} (which is relatively low given the
large number of latent nodes and the high order of the underlying
brain dynamics).  \nrev{It should be noted that the $R$-value is not a
  suitable measure for comparison among the networks obtained by the
  LSAR, S+L, and dDTF methods. On the one hand, the AR model
  underlying the dDTF method is almost identical to the LSAR model
  used here, resulting in almost identical $R$ values, while the
  reconstructed networks are considerably different
  (c.f. Figs.~\ref{fig:eeg-graph-ours}
  and~\ref{fig:eeg-graph-EEGLAB}) due to different interpretations of
  on the model implications for network connectivity. On the other
  hand, the $R$ value is not well-defined for the S+L method since the
  right-hand side of~\eqref{eq:R} is negative, i.e., the reconstructed
  AR model has extremely poor \emph{prediction} performance. This is
  not surprising as the S+L method is aimed at maximizing the entropy
  (and thus minimizing predictability).

  Next, we analyzed the effect of the choice of manifest nodes on the
  reconstructed network. In addition to selecting the 13
  most anterior cortical nodes as above, we performed other runs where
  we selected the 13 most posterior nodes and 13 random nodes to
  reconstruct the manifest network using the LSAR method. We show in
  Fig.~\ref{fig:eeg-loc-R} these node choices (a), the reconstructed
  network for the posterior (b) and random (d) selections ($\alpha =
  0.12$), and (c) the $R$ values for all three cases. Interestingly, the
  density of direct connections is significantly higher among the
  posterior nodes. Also, the LSAR prediction performance is
  significantly lower in this case, suggesting less
  conformity of the occipito-parietal cortex to the simplifying
  assumptions of our AR model (linearity and passivity of latent
  nodes). Consistently, the network density and $R$ value
  of the random case interpolates between the anterior and posterior
  cases, as expected.

  Finally, an interesting observation in Fig.~\ref{fig:eeg-loc-R}(c)
  is that, even an AR model with $\tau = 2$ can predict about $95\%$
  of unseen data in all cases. This, at first glance, questions the
  need for any higher-order models as far as prediction is
  concerned. Nevertheless, notice that even an AR model with $\tau =
  1$, corresponding to an \emph{isolated} manifest subnetwork, can
  predict $90\%$ of unseen data, while the visual discrimination
  task performed by the subject heavily relies on coordination between
  posterior (visual) and anterior (motor planning and execution)
  areas. The reason why this model can predict unseen data so well is
  in the strong dominance of first-order local dynamics of every area
  (the diagonal of $\tilde A_0$) over the rest of network
  dynamics.%
  \footnote{\nrev{This can be easily seen by inspecting the AR coefficients
    $\tilde A_i$ estimated from data, and is physiologically justified
    as each area is composed of millions of neurons that are locally
    densely connected and serve specific purposes but only
    (relatively) sparsely connected with remote areas.}}  Thus,
  the prediction performance of a first-order model serves as a
  \emph{baseline} for higher orders, capturing the contribution
  of local interactions to the overall brain dynamics. This enlightens
  why the $\sim 1\%$ improvement in prediction performance as we go
  from $\tau = 2$ to $\tau = 15 \sim 20$ is significant.}

\begin{figure}
  \subfigure[]
      {\includegraphics[width=0.51\linewidth]{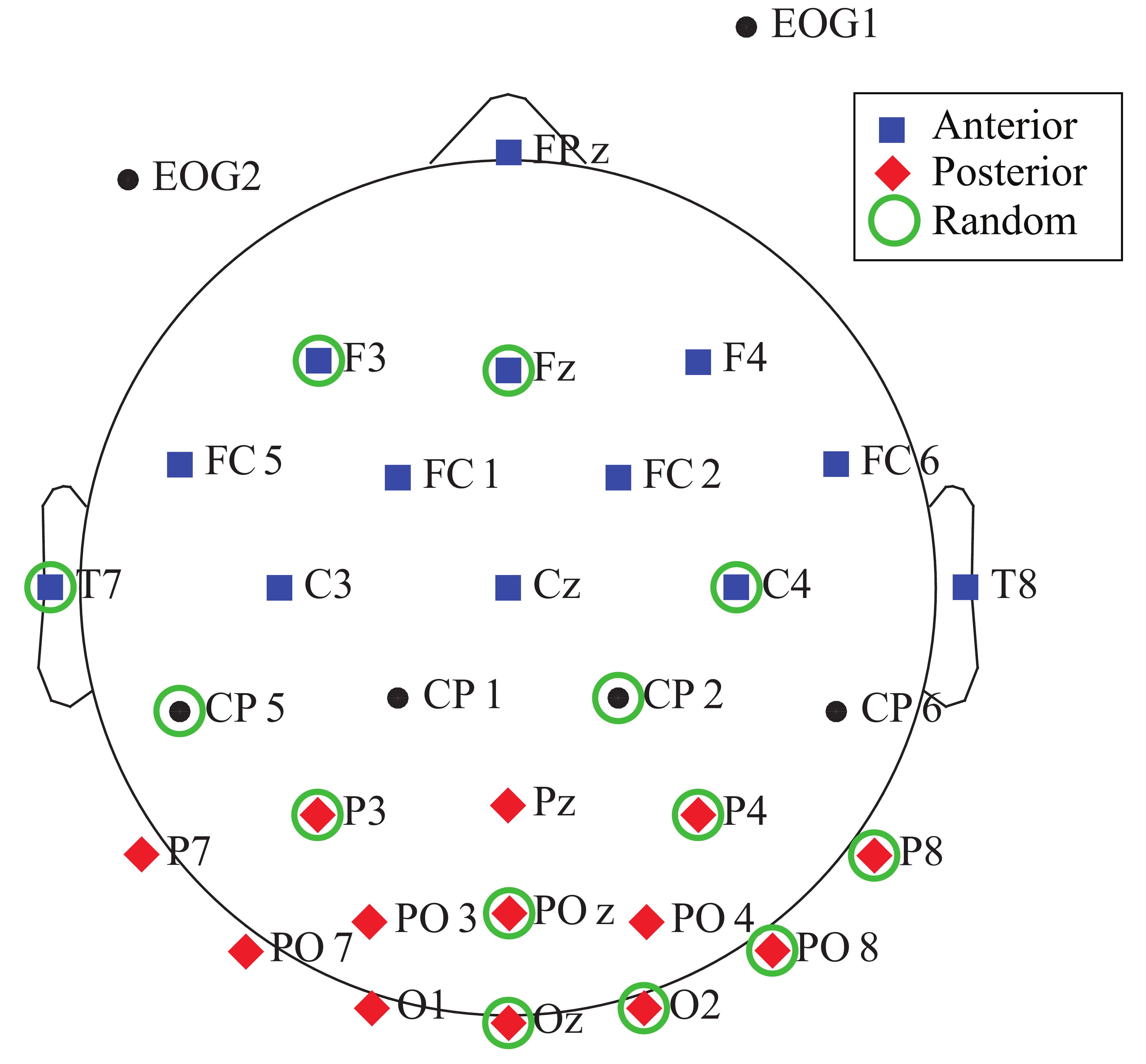}}
      \hspace{1pt} \subfigure[]
      {\includegraphics[width=0.46\linewidth]{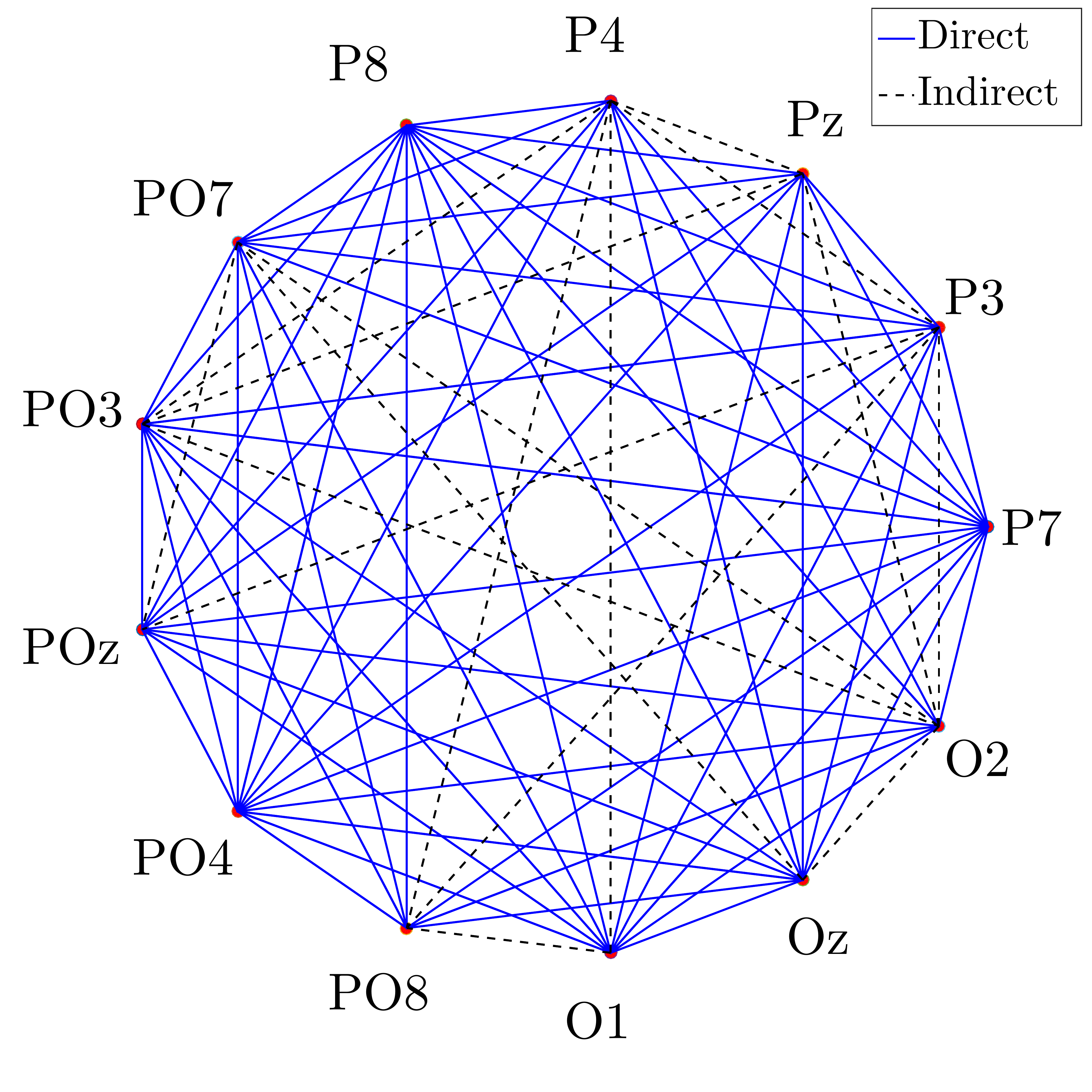}}
      \\
      \subfigure[]
      {\includegraphics[width=0.48\linewidth]{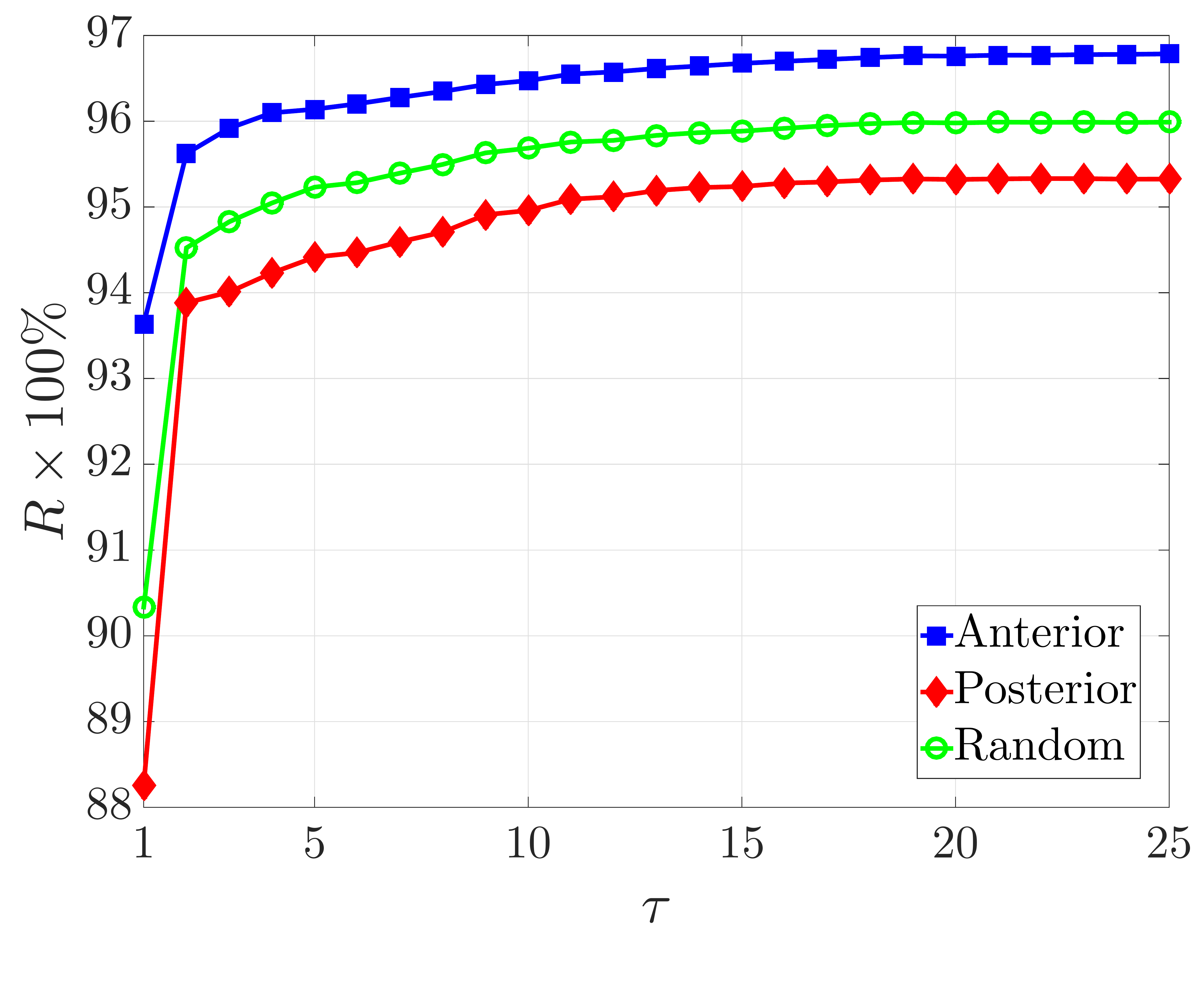}}
      \hspace{5pt}
      \subfigure[]
      {\includegraphics[width=0.46\linewidth]{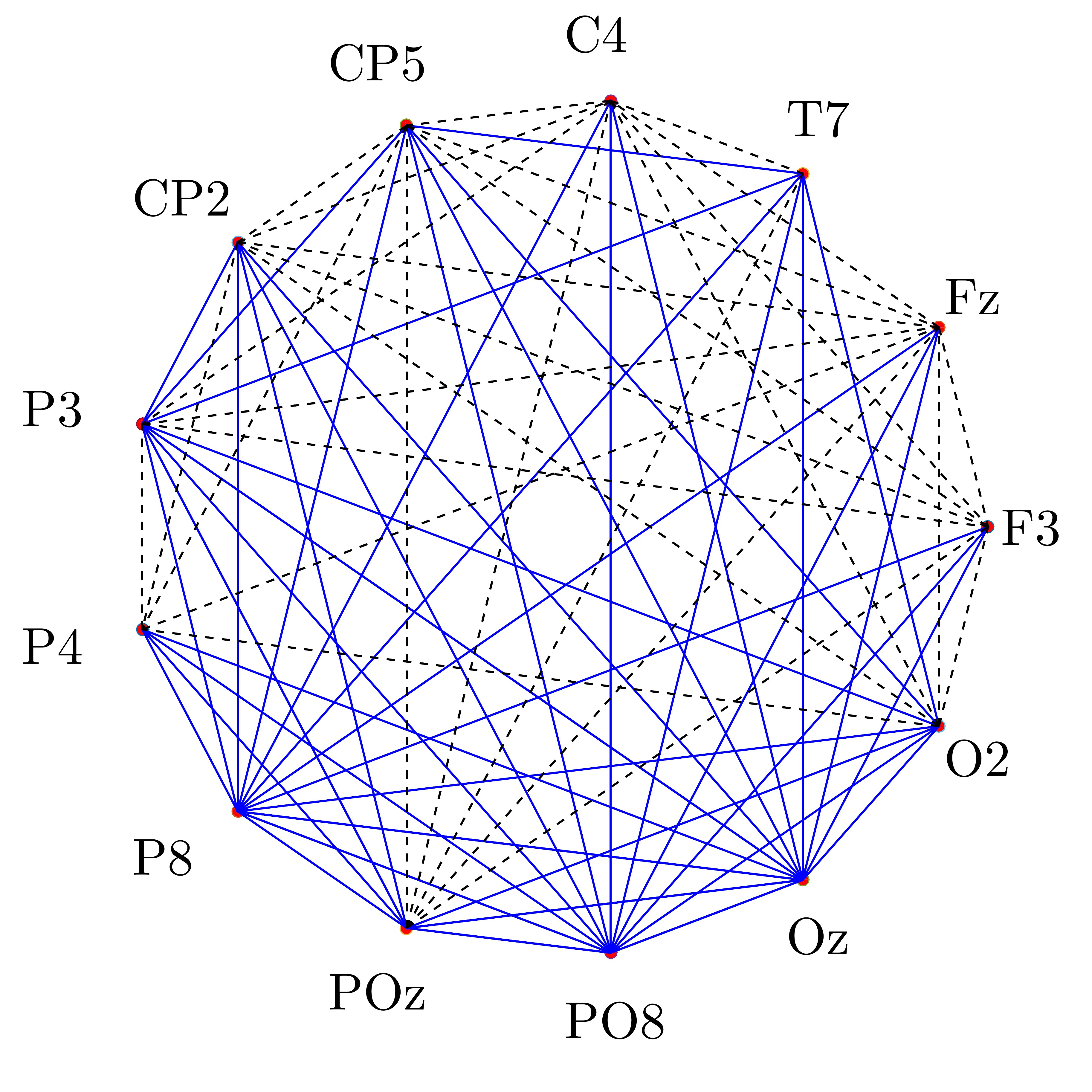}}
      \vspace{-10pt} \caption{\nrev{Comparison between different
          selections of manifest nodes in Example~\ref{ex:eeg}:
          (a) Electrode locations.
          (b and d) The reconstructed network for
          the 13 posterior nodes and 13 random nodes, resp.
          ($\alpha = 0.12$). (c) Prediction performance $R$
	  for the three different choices of manifest
          nodes (reconstructed network for anterior selection is given in
          Fig.~\ref{fig:eeg-graph-ours}(b)).}}
    \label{fig:eeg-loc-R}
    \vspace{-10pt}
\end{figure}
\end{example}
}

\section{Conclusions}\label{Sec:con}

We have considered the problem of identifying the interaction
structure among a group of nodes, termed manifest, that can be
directly \nrev{actuated and measured}, and are part of a larger
linear-time invariant network containing an unknown number of latent
nodes. We have shown that, if there are no inputs to the latent nodes,
then the transfer function of the manifest subnetwork can be
approximated to any degree of accuracy by means of an auto-regressive
model.  We have proposed a least-squares estimation method that uses
measured data to generate estimates that converge in probability to
this AR model exponentially fast as the length of data and the model
order increase. The estimation method does not require any knowledge
of the number or the states of the latent nodes.  We have illustrated
our results in a directed ring network, a group of
Erd\H{o}s--R{\'{e}}nyi random graphs, \rev{and on a time-series of EEG
  data recorded from the human brain.}  Future work will investigate
the sensitivity of the estimation's performance to latent nodes, the
characterization of particular network structures which are easier or
more difficult to identify, the application of our results to the
analysis of brain data, and the extension of the results to network
models \nrev{where, in addition to manifest and latent, there are
  nodes that can be actuated but not measured, and nodes that can be
  measured but not actuated.}

\section*{Acknowledgments}
The authors would like to thank Dr. John Iversen for numerous
discussions, helpful feedback on the contents of the paper, and
introducing us to the SIFT Matlab toolbox.  This work was supported by
NSF Award CNS-1329619.

\vspace{-5pt}
\bibliographystyle{ieeetr}

\vspace{-5pt}
\appendix

\begin{lemma}\label{Lemma1}
  Given two vectors $a,b\in \mathbb{R}^{n}$, it holds for any $ M\in
  \mathbb{R}_{> 0}$ that $\rVert ab^{T}\rVert _{\max }\leq
  M^{-1}a^{T}a+Mb^{T}b$.
\end{lemma}
\begin{proof}
  By definition of the max norm,
  \vspace*{-1ex}
  \begin{align*}
    \rVert ab^{T}\rVert _{\max }& =\max_{1\leq i,j\leq n}\left\vert
      a_{i}b_{j}\right\vert \leq \sum_{i=1}^{n}(M^{-1}\left\vert
      a_{i}\right\vert ^{2}+M\left\vert b_{i}\right\vert ^{2})
    \\
    & =M^{-1}a^{T}a+Mb^{T}b.
  \end{align*}
  \vspace*{-1ex}
\end{proof}

\end{document}